\documentclass[aps,prd,twocolumn,groupedaddress,nofootinbib]{revtex4-1}

\usepackage{graphicx,color,amsmath,mathtools,amssymb}
\usepackage{tcolorbox}

\newcommand{\nocontentsline}[3]{}
\newcommand{\toclesslab}[3]{\bgroup\let\addcontentsline=\nocontentsline#1{#2\label{#3}}\egroup}
\newcommand{\tocless}[2]{\bgroup\let\addcontentsline=\nocontentsline#1{#2}\egroup}

\usepackage{amsthm}

\theoremstyle{definition}
\newtheorem{fact}{Fact}

\newtheorem{prop}{Proposition}

\usepackage{soul}

\newcommand{\OO}{{\mathcal{O}}}

\DeclareMathOperator{\tr}{tr}

\DeclareMathOperator{\Real}{Re}
\DeclareMathOperator{\Var}{Var}

\definecolor{mypurple}{RGB}{164,64,214}
\definecolor{mypink}{RGB}{255, 0, 213}

\usepackage{hyperref} 
\hypersetup{
    colorlinks=true,
    citecolor=blue,
    linkcolor=magenta,
    filecolor=magenta,
    urlcolor=blue,
}

\newcommand{\F}{\mathcal{F}}
\renewcommand{\L}{\mathcal{L}}
\newcommand{\wH}{\widetilde{H}}
\newcommand{\eqg}{\overset{0}{=}}

\newcommand{\ket}[1]{| #1 \rangle}
\newcommand{\bra}[1]{\langle #1 |}

\makeatletter
\def\l@subsubsection#1#2{}
\makeatother

\begin{document}

\title{New bounds on adaptive quantum metrology under Markovian noise}

\author{Kianna Wan}
\email{kianna@stanford.edu}
\affiliation{Stanford Institute for Theoretical Physics, Stanford University, Stanford, CA 94305, USA}
\author{Robert Lasenby}
\email{rlasenby@stanford.edu}
\affiliation{Stanford Institute for Theoretical Physics, Stanford University, Stanford, CA 94305, USA}

\date{\today}

\begin{abstract}
	We analyse the problem of estimating a scalar parameter $g$ that controls the Hamiltonian of a quantum system subject to Markovian noise.
	Specifically, we place bounds on the growth rate of the quantum Fisher information with respect to $g$, in terms of the Lindblad operators and the $g$-derivative of the Hamiltonian $H$.
	Our new bounds are not only more generally applicable than those in the literature---for example, they apply to systems with time-dependent Hamiltonians and/or Lindblad operators, and to infinite-dimensional systems such as oscillators---but are also tighter in the settings where previous bounds do apply.
	We derive our bounds directly from the master equation describing the system, without needing to discretise its time evolution.
	We also use our results to investigate how sensitive a single detection system
	can be to signals with different time dependences.
	We demonstrate that the sensitivity bandwidth is related to the quantum fluctuations of $\partial H/\partial g$, illustrating how `non-classical' states can enhance the range of signals that a system is sensitive to, even when they cannot increase its peak sensitivity.
\end{abstract}

\maketitle

{\hypersetup{linkcolor=blue}
\tableofcontents
}

\section{Introduction}

Many metrological problems take the form of detecting
a small, effectively classical influence acting on
a system. Examples include standard tasks
such as radio wave detection, magnetometry,
and accelerometry, as well as problems in fundamental physics
such as gravitational wave detection
and searches for new forces.

For real-world apparatuses, the sensitivity to small
signals is often constrained by practical noise sources
such as vibrations, or by sources of uncertainties such
as fabrication errors.
However, in some cases, it is possible to suppress
these issues to the point where the fundamental limits imposed
by quantum mechanics become important. For example,
at high enough frequencies, the dominant noise
source in the LIGO gravitational wave
detectors is quantum shot noise~\cite{Saulson1994}.
In order to determine how best to conduct
measurements, it is therefore important
to understand the quantum mechanical limits on signal detection.

If our detector system is under perfect
control, then there is a simple
answer to this question. We will model
a `classical' influence acting on a quantum
system as a c-number parameter that affects
the system's Hamiltonian $H$. 
For a single real parameter $g$,
the distinguishability of different
values is set by the
`quantum Cram\'er-Rao bound' (also referred
to as the
`fundamental quantum limit', or FQL)~\cite{10.1038/nphoton.2011.35},
in terms of the quantum fluctuations of
$H' \equiv \partial H/ \partial g$.
The FQL has been used to analyse
systems including gravitational
wave detection experiments~\cite{1608.00766,1903.09378}
and axion dark matter detection experiments~\cite{10.1103/PhysRevD.103.075007},
providing insights into their sensitivity limits.

However, in many cases, a detector
system is coupled to a complicated environment,
which can prevent us from reaching the FQL sensitivity.
A range of papers in the quantum metrology
literature have investigated sensitivity limits under various different assumptions about the environment
and its coupling to the detector system.

A common and useful approximation takes the
environment to be a Markovian bath, in which case
the evolution of the detector system is governed
by a Lindblad master equation~\cite{10.1017/CBO9780511813948}. This was
examined in a number of recent papers~\cite{10.22331/q-2017-09-06-27,prx,preskill,10.1103/PhysRevResearch.2.013235}, which
considered parameter estimation for a time-independent
$H'$ acting on a finite-dimensional detector
system. They found that if $H'$ can be written as a suitable
combination of the Lindblad operators associated
with the coupling to the environment, then the
sensitivity to $g$---quantified by the
quantum Fisher information (QFI)---scales at best
linearly over large times. Conversely, if this
condition is not satisfied, the QFI can grow as $\sim\! t^2$.

In this paper, we analyse Hamiltonian parameter
estimation for a system evolving according
to a Markovian master equation, without
assuming time-independence or finite dimensionality.
We derive bounds on the growth rate of the
QFI; in the case of finite-dimensional systems
with time-independent $H'$ and Lindblad operators,
these sharpen the previously derived limits.
Our new bounds also apply more broadly, including to
infinite-dimensional systems such as oscillators,
and to time-dependent $H'$ and Lindblad
terms. Similarly to~\cite{10.22331/q-2017-09-06-27,prx,preskill,10.1103/PhysRevResearch.2.013235}, they
encompass detection schemes involving general
adaptive strategies.

We arrive at our bounds using different methods from those in prior work. Unlike previous analyses, our derivation does not rely on discretising the time evolution
of the system. Instead, we use symmetric logarithmic derivatives
to directly bound the time derivative of the QFI
starting from the master equation, enabling
a fully time-dependent treatment.
We compare our results in detail to existing bounds, 
showing that ours are tighter even in the restricted settings where the latter apply.

Being able to analyse time-dependent signals
allows us to address additional metrological questions.
The QFI assumes a signal that depends on a single
real parameter $g$; in particular, the time dependence
of the signal (for given $g$) is assumed
to be known. However, in many situations,
we are interested in a whole
range of possible signals, with different time dependences.
Examples include signals with \textit{a priori} unknown
frequency, such as axion
dark matter with an unknown mass~\cite{10.1146/annurev-nucl-102014-022120},
or gravitational waves from a pulsar of unknown
spin rate~\cite{10.1017/pasa.2015.35}.

In such cases, we are interested in both
the peak sensitivity at the optimal frequency,
and the bandwidth over which we
can (approximately) attain this sensitivity.
We demonstrate that, in a range of different
situations, the sensitivity bandwidth can
be related to the short-time growth
rate
of the QFI, which itself is related
to the quantum fluctuations of $H'$.
Consequently,
using states with enhanced fluctuations
can increase the sensitivity bandwidth,
even when the noise terms prevent this from
increasing the peak sensitivity.

For example, if we consider the problem of
near-resonant force detection---detecting
a small classical forcing acting on a damped
harmonic oscillator---then there is a state-independent
bound on the sensitivity. This can be attained,
for an on-resonance forcing, by a critically-coupled
oscillator
in its ground state.
Using `non-classical' states, such as squeezed
states or Fock states, cannot improve this peak sensitivity
(assuming Markovian damping). However, such states
can have larger $H'$ fluctuations, and correspondingly,
can enhance the bandwidth over which near-peak sensitivities
can be achieved. This behaviour has been noted
for schemes using squeezed coherent states~\cite{konrad}---we show
that it applies more generally.

\section{Hamiltonian parameter estimation}
\label{secqfi}

We will suppose that we have some quantum system that can be described by a master equation in Lindblad form~\cite{nc,10.1017/CBO9780511813948}, 
\begin{equation}
	\dot \rho = -i [H,\rho] + \sum_j
	\Big(L_j \rho L_j^\dagger - \frac{1}{2}\left\{L_j^\dagger L_j, \rho\right\}\Big),
	\label{eqsme}
\end{equation}
where $\rho$ is the system's density operator,
$H$ is its Hamiltonian, the $L_j$
are Lindblad operators describing its interaction
with a Markovian environment, and $\dot{\rho} \equiv \partial \rho / \partial t$ is the time derivative of $\rho$.  
We assume that $H$ depends on some c-number parameter
$g \in \mathbb{R}$ that we are trying
to determine.
For instance, $g$ might correspond to the strength of the
signal. A simple but illustrative setup of this kind is a two-level system
with a $g$-dependent energy splitting,
$H = g \epsilon \sigma_z$, subject to dephasing noise, which is described by a single Lindblad operator $L_1 = \sqrt{\gamma} \sigma_z$:
\begin{equation}
	\dot \rho = - i g\epsilon[\sigma_z,\rho] + \gamma (\sigma_z \rho \sigma_z - \rho).
	\label{eq_tlsdephase}
\end{equation}

Henceforth, all operators are functions of $t$ and $g$ in general, except for the Lindblad operators, which depend only on $t$.\footnote{More generally, one could also consider
$g$-dependent Lindblad operators,
corresponding to parameters of the environment
and its coupling that are not known~\cite{10.1103/PhysRevA.72.052334,10.1088/0305-4470/39/46/015,10.1088/1367-2630/15/7/073043,1611.09165,10.1103/PhysRevLett.118.100502,prx}; we discuss
this case in Appendix~\ref{sec_lindblad_app}.} We suppress the $t$- and $g$- dependence of operators as well as scalar quantities in many equations; unless noted otherwise, these equations hold for all values of $t$ and $g$. Throughout, we use over-dots to denote differentiation with respect to $t$, and primes to denote differentiation with respect to $g$. 

Given $\rho(t_0,g)$ as a function of $g$ for a fixed time $t_0$, the master equation (Eq.~\eqref{eqsme}) tells us the state $\rho(t,g)$ at later times $t$.
The distinguishability of different $g$ values can be quantified via the quantum Fisher information (QFI), which is the maximum, over different measurements,
of the classical Fisher information (with
respect to $g$) of the measurement
outcome~\cite{10.1103/PhysRevLett.72.3439}.
The quantum Cram\'er-Rao bound (QCRB) states
that the uncertainty $\delta g$ in measuring $g$
is bounded as $\delta g \ge 1/\sqrt{n \F}$,
where $\F$ is the QFI and $n$ is the number
of repetitions of the experiment
(specifically, the QCRB lower-bounds the variance of any unbiased
estimator of $g$; this bound can always be attained in the large-$n$ limit)~\cite{10.1103/PhysRevLett.72.3439}.

\subsection{General QFI bounds}
\label{sec_qfibounds}

For given $\rho = \rho(t,g)$, a useful quantity to define is the `symmetric logarithmic derivative' (SLD) $\mathcal{L} = \mathcal{L}(t,g)$,  which is a Hermitian operator
with the property that
\begin{equation} \label{eq_ldefn}
	\rho' = \frac{1}{2}\left(\rho \L + \L \rho\right).
\end{equation} (As we show in
Appendix~\ref{appL}, if it is possible to find $\L(t_0,g)$ at some time $t_0$, then it is always possible to find $\L(t,g)$ at later times.)
The QFI is then given by~\cite{10.1103/PhysRevLett.72.3439}
\begin{equation}
    \F = \tr\left(\rho \L^2\right).
\end{equation}
Under the assumption that $\rho, \rho'$ and $\L$, are differentiable with respect to time,\footnote{More generally, as we discuss in Appendix~\ref{appL}, it is sufficient for our purposes that $\rho,\rho'$, and $\L$ are differentiable with respect to $t$ (for given $g$) except possibly at a set of isolated points in time, as long as $\F$ is continuous at these points (since then $\F$ can still be bounded by integrating $\dot \F$ over the intervals between the isolated points).  We show in Appendix~\ref{appL} that both of these criteria are fulfilled under appropriate regularity conditions.}
\begin{equation}
	\dot \F = \tr \left(\dot \rho \L^2 + \rho \dot{\L} \L + \rho \L \dot{\L}\right),
\end{equation}
which can be written using Eq.~\eqref{eq_ldefn} as~\cite{10.1103/PhysRevA.82.042103}
\begin{equation} \label{eq_dfq2}
	\dot \F = 2 \tr \left(\dot \rho' \L\right) - \tr\left(\dot{\rho} \L^2\right).
\end{equation}
This form allows us to use the master equation expression
for $\dot \rho$ (Eq.~\eqref{eqsme}); after
some algebra, we obtain
\begin{equation}	\label{eqfqdot}
	\dot \F = 2 i \tr\left(\rho [H',\L]\right)
	- \sum_j \tr \left(\rho [L_j, \L]^\dagger [L_j,\L]\right).
\end{equation}
Note that this expression depends only on $H'$, and not on $H$. This is to be expected,
since $g$-independent {evolution} cannot contribute
to distinguishing between different values of $g$.

Next, to bound $\dot{\F}$, it is helpful to subtract combinations of the Lindblad operators $L_j$ from $H'$. Specifically, for any scalar coefficients $\alpha \in \mathbb{R}$ and $\beta_j, \gamma_{jk} \in \mathbb{C}$ with $\gamma_{jk} = \gamma_{kj}^*$ (where the indices $j,k$ run over the same range as for the Lindblad operators), we can write
\begin{equation}
	H' = G+ \alpha I + \sum_j (\beta_j^* L_j + \beta_j L_j^\dagger) + \sum_{j,k} \gamma_{jk}
	L_j^\dagger L_k
	\label{eq_hdecomp}
\end{equation}
for some Hermitian operator $G$. (Recall that the $t$- and $g$-dependence of operators is implicit; similarly, the coefficients $\alpha, \beta_j, \gamma_{jk}$ are allowed to vary with $t$ and $g$.)
Then,
\begin{align}
	i [H',\L] &= i[G,\L] + \sum_j \left[
	i \Big(\beta^*_j + \sum_k \gamma_{kj} L_k^\dagger\Big) [L_j,\L] + {\rm h.c.}\right]
	 \nonumber \\
	&= i [G,\L] +
\sum_j	\left( A_j^\dagger [L_j,\L] + {\rm h.c.}\right) \label{iH'L}
\end{align}
where we define
\begin{equation} \label{eq:A_j} A_j \coloneqq i \left(\beta_j I + \sum_k \gamma_{jk} L_k\right). \end{equation} Substituting Eq.~\eqref{iH'L} into Eq.~\eqref{eqfqdot}, we have
\begin{align}
    \dot{\F} &= 2i\tr(\rho[G,\L]) + 4\sum_j\mathrm{Re}\left[\tr(\rho A_j^\dagger [L_j, \L])\right] \nonumber \\
    &\quad - \sum_j \tr\left(\rho[L_j,\L]^\dagger[L_j,\L]\right) \\
    &\leq 2i\tr\left(\rho[G,\L]\right) \nonumber \\
    &\quad + 4\sum_j \left\{\sqrt{\tr(\rho A_j^\dagger A_j)} \sqrt{\tr\left(\rho[L_j,\L]^\dagger [L_j,\L]\right)} \right.\nonumber\\
    &\quad - \left. \frac{1}{4} \tr\left(\rho[L_j,\L]^\dagger [L_j,\L]\right)\right\} \label{cauchyschwarz}\\
    &\leq 2i \tr(\rho[G,\L]) + 4\sum_j \tr(\rho A_j^\dagger A_j). \label{quadratic_bound}
\end{align}
Here, the first inequality follows from applying the Cauchy-Schwarz inequality to the Hilbert-Schmidt inner product $\langle A, B\rangle = \tr(A^\dagger B)$ (with $A = A_j\sqrt{\rho}$ and $B = [L_j,\L]\sqrt{\rho})$. To obtain the second inequality, we use the fact that for any $a \geq 0$, the function $f(x) = \sqrt{ax} - x/4$ is maximised at $x = 4a$, with value $f(4a) = a$.

For general $G$, we cannot obtain an $\L$-independent bound for
the $2 i \tr(\rho [G,\L])$ term, but we can bound it in terms of $\F$ as follows.
For any Hermitian operator $A$, we define a Hermitian operator $\L_A$ such that
\begin{equation} \label{eqLG} i [\rho, A] = \frac{1}{2}(\rho \L_A + \L_A \rho)
\end{equation}
(this is always possible, as shown in Appendix~\ref{appqfi}). Then,
\begin{align}
	2 i \tr (\rho [G,\L])
	 &= \tr (\L (2 i [\rho,G])) \nonumber \\
	 &= \tr (\L (\rho \L_G + \L_G \rho)) \nonumber \\
	 &= 2 \Real \tr (\rho \L_G \L) \nonumber \\
	 &\le 2 \sqrt{\tr(\rho \L_G^2)} \sqrt{\tr(\rho \L^2)} \nonumber \\
	 &= 2 \sqrt{\F_G \F}
\end{align}
where we write
\begin{equation} \label{eq:FA}
    \F_A \coloneqq \tr(\rho \L_A^2)
\end{equation}
for the quantum Fisher information with respect to a general Hermitian operator $A$, as defined in~\cite{10.1088/1751-8113/47/42/424006}.\footnote{In~\cite{10.1088/1751-8113/47/42/424006}, $\F_A$ is denoted by $F_Q[\rho,A]$.}
Thus, we can bound the rate of increase
$\dot{\F}$ (at a given $t$, $g$) by
\begin{align} \label{eqlim0}
	\dot \F \le 4 \left(\sqrt{\frac{\F_G}{4}\F}
	+ \sum_j \langle A_j^\dagger A_j\rangle \right)
\end{align}
for any $G$ and $A_j$'s satisfying Eqs.~\eqref{eq_hdecomp} and~\eqref{eq:A_j}, where angled brackets denote
the expectation value in the state~$\rho$, $\langle A \rangle \coloneqq \tr(\rho A)$.
Eq.~\eqref{eqlim0} is saturated iff
$[L_j,\L] \sqrt\rho = 2 A_j \sqrt\rho$ for all $j$
and $\L \sqrt \rho = c \L_G \sqrt \rho$
for some $c \ge 0$ (or $\L_G \sqrt\rho = 0$).
The RHS depends on the coefficients $\beta_j, \gamma_{jk}$
in Eq.~\eqref{eq_hdecomp} (note that the coefficient $\alpha$ has no effect, since $\F_{G + \alpha I} = \F_G$ for any $\alpha$, cf.~Fact~\ref{fact:FaA+bI} in Appendix~\ref{appqfi}), so we obtain the tightest bound by considering the minimum over possible values. This gives us our main result:
\begin{tcolorbox}
	\begin{equation}
		\dot \F \le 4 \min_{\substack{\beta_j,\gamma_{jk} \in \mathbb{C} \\ \gamma_{jk} = \gamma_{kj}^*}}
		\left(\sqrt{\frac{\F_G}{4}\F}
	+ \sum_j \langle A_j^\dagger A_j \rangle \right)
		\label{eqlim1}
	\end{equation}
	for any $t$ and $g$, where for each choice of $\beta_j$ and $\gamma_{jk}$, $G$ and $A_j$ are defined as in Eqs.~\eqref{eq_hdecomp} and~\eqref{eq:A_j}, and $\F_G$ is defined as in Eq.~\eqref{eq:FA}.
\end{tcolorbox}
As presented, Eq.~\eqref{eqlim1} may seem
somewhat abstract.
For any values of $t$ and $g$,
and operators $\rho$ and $H'$ at those
values, one could perform the optimisation
over $\beta_j, \gamma_{jk}$ computationally (for given
$\F$),
but the minimum does not have an analytic
form in general.
Moreover, as we discuss in Section~\ref{sec_control},
it may not be possible to saturate the resulting bound.
However, as we show in subsequent sections, even simple, potentially non-optimal
choices of $\beta_j, \gamma_{jk}$ can give useful
constraints on $\dot{\F}$ and hence $\F(t)$, which enable us to tighten bounds in the existing
literature; see Section~\ref{sec_compare}. Furthermore, for particular
systems of interest, it is often straightforward to optimise
$\beta_j, \gamma_{jk}$, as illustrated
in Section~\ref{secqho}.
The power of Eq.~\eqref{eqlim1} is that it
enables one to bound the QFI
growth rate in terms of $H'$,
without reference to the full Hamiltonian $H$, which (as
we discuss in Section~\ref{sec_control})
may be arbitrarily complicated.
In addition, Eq.~\eqref{eqlim1} naturally allows for 
$t$- and $g$-dependence; the optimal choices
for $\beta_j$, $\gamma_{jk}$ may
vary with $t$ and $g$.

Since $\F_G$ is convex in $\rho$~\cite{10.1088/1751-8113/47/42/424006}, the RHS
of Eq.~\eqref{eqlim0} is convex in $\rho$,
for given $\F$. This means that,
as we would expect, the $\dot \F$
bound for a probability mixture of states is at
most as large as the probabilistic average
of the $\dot \F$ bounds for the states in the mixture.
For pure states,
$\frac{1}{4} \F_G = \langle G^2 \rangle - \langle G \rangle^2
\eqqcolon \Var (G)$ corresponds to the quantum fluctuations
of the operator $G$,
so in general, we have
$\frac{1}{4} \F_G \le \Var(G)$ (cf.~Fact~\ref{fact:FVarA} in Appendix~\ref{appqfi}).

To derive a bound from Eq.~\eqref{eqlim0} or~\eqref{eqlim1}
that applies to a whole class of states, as opposed
to some specific $\rho$, we need
to bound
$\F_G$ and $\sum_j \langle A_j^\dagger A_j \rangle$ 
for that class.
In some circumstances, this is very simple;
for example, if we take $\gamma_{jk} = 0$ for all $k$,
then $\langle A_j^\dagger A_j \rangle = |\beta_j|^2$,
independent of the state.
More generally, we have the state-independent bounds
\begin{equation}
	\langle A_j^\dagger A_j \rangle
	\le \Big\|A_j^\dagger A_j\Big\|,
	\qquad
	\frac{1}{4}\F_G \le \Var(G) \le \min_{x\in \mathbb{R}} \| G - x I \|^2
\end{equation}
(where $\|\cdot\|$ denotes the operator norm), which are well-defined whenever the operators $A_j$ and $G$ are bounded, e.g., for any finite-dimensional system.
Tighter bounds may be possible for more restricted
classes of states
(and if our operators have unbounded norm,
as can occur in infinite-dimensional systems, then
state-independent bounds will not exist).
For instance, in the case of a damped
harmonic oscillator with quadratic forcing,
as we discuss in Section~\ref{secqho},
the best large-$\F$ bound comes from $\sum_j A_j^\dagger A_j \propto N$, where
$N$ is the oscillator's number operator.
This operator is unbounded on the full Hilbert space,
but if the expectation value of the oscillator's energy
is bounded, then we can bound $\sum_j \langle A_j^\dagger
A_j\rangle$. Note that we do not necessarily need
to restrict to a finite-dimensional
 subspace, as done in~\cite{preskill}; in the previous example,
the set of states with $\langle N \rangle < \overline N$
is not contained within any finite-dimensional subspace, for any $\overline N > 0$.

\subsection{Quantum control and adaptive protocols}
\label{sec_control}

By interpreting our master equation appropriately,
we will see that our bounds can encompass
general measurement strategies, including
adaptive procedures.
In many cases, our detection
system consists of a `probe' subsystem on which 
$H'$ and $L_j$ act, along with other degrees of freedom
such as ancillae. To take into account
interactions of the probe with the ancillae,
we will consider our master equation to describe the state $\rho$ of
the entire detection system, excluding
the `environment'
degrees of freedom that are responsible for the 
Lindblad terms in the master equation. Effectively, we can think of our
detection system as representing the degrees of freedom
we have good control over (as opposed
to the Markovian environment). Hence,
the Hamiltonian $H$ can include
arbitrary ($g$-independent)
interactions of the probe with ancillae
degrees of freedom.
Since Eq.~\eqref{eqlim1}
does not depend on the $g$-independent
part of $H$, the inclusion of such interactions
does not affect our $\dot \F$ bounds.
While the master equation assumes that $\rho$ evolves
smoothly in time, there is no issue with taking the evolution
to be very fast (e.g., to model
instantaneously applied quantum gates\footnote{Alternatively,
we can observe that instantaneous $g$-independent
operations cannot affect the QFI,
and $\dot \F$ obeys our bounds both before and after
such operations.}), as long
as the assumption of Markovian noise
is not violated.

We can also model measurement-induced non-unitarity,
including adaptive procedures,
via the principle of deferred measurement~\cite{nc,10.1103/PhysRevLett.106.090401}.
That is, any procedure with intermediate measurements and classically controlled feedback is equivalent to a unitary procedure with all measurements performed at the end. 
Thus, our analysis
allows for `full and fast quantum control'
of the kind assumed
in the literature~\cite{10.22331/q-2017-09-06-27,prx,preskill}.

It is clear that, for some systems,
the bound in Eq.~\eqref{eqlim1} cannot be
saturated, for large enough $\F$.
For example, consider a single spin with time-independent $H'$ and Lindblad
operators $L_1 \propto \sigma_x$,
$L_2 \propto \sigma_y$,
$L_3 \propto \sigma_z$. Then, as we take $\F = \tr(\rho \L^2)$ large, one or more of the $[L_j, \L]\sqrt{\rho}$ becomes arbitrarily large, while $2A_j\sqrt{\rho}$ remains bounded for the optimal choice of $\beta_j, \gamma_{jk}$ in Eq.~\eqref{eqlim1}, which sets $G = 0$. Therefore, we cannot saturate Eq.~\eqref{eqlim1} in this case---indeed, we obtain large and negative
$\dot \F$ for large enough $\F$.

To avoid this issue, the information corresponding
to large $\F$ needs to be stored somewhere
that is not affected by the Lindblad terms.
If we split our system into a probe
system,
along with noiseless ancillae, then
the information can simply live in the ancilla
degrees of freedom. A natural question is then
whether, for a fixed probe system with given
$H'$ and $L_j$, one can always find
an ancilla-assisted scheme
that saturates our growth rate bound.
In some situations, we can exhibit specific
detection schemes which do so (as in Section~\ref{secqho}).
We leave the question of
whether, for arbitrary $H'$ and $L_j$, one can find
an extended system (i.e., probe and ancillae) with
$\rho$ and $\L$ saturating Eq.~\eqref{eqlim1}
for any $\F$,
to future work.\footnote{This would be a stronger
statement than those proved in e.g., \cite{preskill,10.1103/PhysRevResearch.2.013235}, which showed that, in the limit of 
large $\F$, it is possible to find an extended system
which is asymptotically close to saturating the bound.}

\subsection{Initial QFI growth}
\label{secshort}

To obtain bounds on $\F(t)$, given some initial
conditions, we can integrate
Eq.~\eqref{eqlim1}. For example,
if we prepare our system in some $g$-independent
state at $t = 0$, then $\rho'(t=0) = 0$ and we can set $\L(t = 0) = 0$,
so $\F(t = 0) = 0$. Then, by Eq.~\eqref{eqlim1},
\begin{equation}
	\int_0^{\F(t)} \frac{d\F}{2 \sqrt{\F_G \F} + 4 \sum_j
	\langle A_j^\dagger A_j\rangle}
	\le t
\end{equation}
for any choice of $\beta_j,\gamma_{jk}$ in the
denominator.
By Eq.~\eqref{eqfqdot}, $\dot \F(t=0) = 0$,
so $\F(t) = \OO(t^2)$ for small $t$.
Therefore, for small enough $t$, the denominator
is minimised by taking
$\beta_j, \gamma_{jk} = \OO(t)$.
With this choice, $\F_G = \F_{H'}(0) + \OO(t)$, where $\F_{H'}(0) \coloneqq \F_{H'}(t=0)$, 
and we have
\begin{equation}
	\frac{\sqrt \F}{\sqrt{\F_{H'}(0)}}(1 + \OO(t)) \leq t
	\quad \Rightarrow \quad
	\F(t) \le \F_{H'}(0) t^2 (1 + \OO(t)),
\end{equation}
so the short-time behaviour of $\F$ is bounded
in terms of $\F_{H'}(0)$.

In fact, the initial growth of $\F(t)$
attains this bound to leading order in $t$.
To see this, note that for small $\L$, corresponding to small
$\F$, differentiating
the expression for $\dot \F$ in Eq.~\eqref{eqfqdot}
with respect to time gives
\begin{equation}
	\ddot \F \simeq 2 i \tr(\rho[H',\dot \L])
	\label{eqfqddot}
\end{equation}
since the other terms are suppressed for small $\L$.
Similarly,
\begin{equation}
	\dot \rho' \simeq \frac{1}{2}(\rho \dot \L + \dot \L \rho)
\end{equation}
From the master equation (Eq.~\eqref{eqsme}),
the only unsuppressed term in the expression for $\dot \rho'$
is
$\dot \rho' \simeq - i [H',\rho]$,
so by Eq.~\eqref{eqLG}, we can replace
the $\dot \L$ term in Eq.~\eqref{eqfqddot}
by $\L_{H'}$.
Hence, for small $t$,
\begin{equation}
	\ddot \F \simeq 2 i \tr (\rho [H', \L_{H'}])
	= 2 \tr (\rho \L_{H'}^2) = 2 \F_{H'}
	\label{eq_fqinitial}
\end{equation}
so $\F(t) \simeq \frac{1}{2}\ddot \F(t=0) t^2 = \F_{H'}(0)t^2$. For pure initial
states, this gives $\F(t) \simeq 4 \Var(H'(t=0)) t^2$,
which corresponds to the QCRB~\cite{10.1038/nphoton.2011.35}.

More generally, this relationship,
for noiseless evolution under a Hamiltonian,
is indeed what motivates the definition
of $\F_{H'}$~\cite{10.1088/1751-8113/47/42/424006}.
We went through the derivation of Eq.~\eqref{eq_fqinitial}
to highlight the approximations
under which this applies in the presence of
Lindblad terms.

\subsection{Comparison to previous results}
\label{sec_compare}

In recent years, a number of papers~\cite{10.22331/q-2017-09-06-27,prx,preskill,10.1103/PhysRevResearch.2.013235} have
investigated the same problem we analyse in this
work: how fast the QFI can grow, for a system
governed by a master equation with a parameter-dependent
Hamiltonian. These papers adopted
a somewhat different approach to ours;
instead of bounding $\dot \F$, they discretised
the time evolution
between $t=0$ and $t=t_{\rm tot}$ into $N$ equal segments,
and used bounds on the QFI from $N$ identical
operations~\cite{10.1103/PhysRevLett.113.250801,10.1088/1751-8113/41/25/255304} to bound $\F(t_{\rm tot})$ in the
limit as $N \rightarrow \infty$.\footnote{These
operations represented
the $g$-dependent part of the Hamiltonian; arbitrary $g$-independent operations
between segments were allowed.}
This approach most naturally applies to
time-independent master equations,
and results in looser bounds than ours, as
we demonstrate below.
In addition, they generally worked
with finite-dimensional systems, which we will
not restrict ourselves to.

We can see from Eq.~\eqref{eqlim1} that the
large-$\F$ scaling of $\dot \F$ depends on whether
$G$ can be set to zero. In Section~\ref{sec_hls},
we compare our bounds to those of
\cite{prx,preskill} in the case where
we can set $G=0$ (for time-independent $H'$).
We show that our bounds have the same
large-time scaling as those from~\cite{prx,preskill},
but are tighter at any finite time.
In the case where $G$ cannot be set to zero, \cite{preskill}
showed that, by using error correction
techniques, one can achieve $\F(t) \sim t^2$ scaling
at large times. We demonstrate in Section~\ref{sec_hnls} that their
schemes have asymptotically optimal scaling,
but that tighter bounds than those derived
in~\cite{prx,preskill} can be placed
on $\F$ at finite times.

\subsubsection{$H'$ in Lindblad span (HLS)}
\label{sec_hls}

The case of time-independent $H'$ and $L_j$
was analysed in~\cite{prx,preskill}.
They demonstrated that if
we can write
\begin{equation}
H' = \widetilde \alpha I + \sum_j (\widetilde\beta_j^* L_j + \widetilde\beta_j L_j^\dagger)
+ \sum_{j,k}\widetilde\gamma_{jk} L_j^\dagger L_k
	\label{eqhls}
\end{equation}
for some coefficients
$\widetilde\alpha, \widetilde\beta_j, \widetilde\gamma_{jk}$---a condition referred to as $H'$
being in the `Lindblad span' (HLS) by~\cite{preskill,zhouthesis}---and we start
with $\rho'(t = 0) = 0$ (and hence $\F(t=0)=0$),
then
\begin{equation}
\F(t) \le 4 \bigg\lVert \sum_j \widetilde A_j^\dagger \widetilde A_j\bigg\rVert t,
	\label{eqlinbound}
\end{equation}
where $\widetilde A_j\coloneqq i (\widetilde\beta_j + \sum_k \widetilde\gamma_{jk} L_k)$.\footnote{In
the notation of~\cite{prx,preskill}, the Hamiltonian
$H$ acts only on the probe subsystem, denoted
by $\mathcal{H}_P$, rather than on the full probe-plus-ancillae
system $\mathcal{H}_P \otimes \mathcal{H}_A$, as in our setup.
However, since the QFI bounds depend only
on $H'$, which acts only on the probe subsystem in both cases,
this does not make any difference.}

Eq.~\eqref{eqlinbound} is an immediate
consequence of our general bound, Eq.~\eqref{eqlim1}.
Moreover, we can see from Eq.~\eqref{eqlim1} that
Eq.~\eqref{eqlinbound} cannot be tight.
Specifically,
Eq.~\eqref{eqlinbound} is the special case of Eq.~\eqref{eqlim0} where $G$ is taken to be $0$,
whereas (as analysed in Section~\ref{secshort})
the optimal choice at small times sets $\beta_j, \gamma_{jk}$ such that
$G = H' + \OO(t)$, which gives
a bound $\F(t) \le \F_{H'} t^2 (1 + \OO(t))$.
Of course, when $t$ is large enough, and $\F$
is correspondingly large,
the best bound (from Eq.~\eqref{eqlim1}) on $\dot \F$ will
approach $4 \sum_j \langle A_j^\dagger A_j \rangle$.
Consequently, if there are no restrictions
on $\rho$, then Eq.~\eqref{eqlinbound}
matches the scaling of Eq.~\eqref{eqlim1} at large $t$.\footnote{As noted in~\cite{prx}, one
can replace
$\lVert \sum_j A_j^\dagger A_j \rVert$ in
Eq.~\eqref{eqlinbound}
with $\sum_j \tr (\rho A_j^\dagger A_j)$ (corresponding
to Eq.~\eqref{eqlim0} with $G=0$),
which can have a tighter bound if one
has extra information about $\rho$.}

As an illustration of how Eq.~\eqref{eqlinbound}
can be sharpened using our approach, suppose
that we take
$\alpha = \delta(t)\widetilde{\alpha}$,\footnote{We make this choice for presentational simplicity---recall that $\alpha$ does not affect the $\dot{\F}$ bound of Eq.~\eqref{eqlim0}, due to Fact~\ref{fact:FaA+bI}.} $\beta_j = \delta(t)\widetilde \beta_j$, and
$\gamma_{jk} = \delta(t)\widetilde \gamma_{jk}$ in Eq.~\eqref{eqlim0},
for some real function $\delta$ whose value we will choose
at each $t$ to obtain the best bound (see Eq.~\eqref{eq:deltachoice}).
Then, we have $G = (1-\delta)H'$, so Eq.~\eqref{eqlim0} becomes (using Fact~\ref{fact:FaA+bI})
\begin{equation}
	\dot \F \le 4\left(|1 - \delta| \sqrt{\frac{\F_{H'}}{4}\F} + \delta^2 \sum_j \langle \widetilde{A}_j^\dagger \widetilde{A}_j\rangle\right)
	\label{eq27}
\end{equation}
Now, supposing that $\frac{1}{4}\F_{H'} \le c_1^2$
and $\sum_j \langle \widetilde{A}_j^\dagger \widetilde{A}_j\rangle \le c_2$
for some non-negative $c_1$ and $c_2$,
\begin{equation}
	\dot \F \le
4 \left(|1 - \delta| c_1 \sqrt{\F}
	+ \delta^2 c_2 \right).
\end{equation}
The RHS is minimised by taking
\begin{equation} \label{eq:deltachoice}
	\delta = \begin{dcases}
		\frac{c_1 }{2 c_2}\sqrt{\F} \qquad &\sqrt{\F} \leq \frac{2 c_2}{c_1}  \\
		1  \qquad &\sqrt{\F} \ge \frac{2 c_2}{c_1}
	\end{dcases}
\end{equation}
so
\begin{equation} \label{eq:fqdothls}
	\dot \F \le \begin{dcases}
		4 c_1 \sqrt{\F}\left(1 - \frac{c_1}{4 c_2}\sqrt{\F}\right) \qquad &\sqrt{\F} \leq \frac{2 c_2}{c_1} \\
		4 c_2 \qquad &\sqrt{\F} \ge \frac{2 c_2}{c_1}.
	\end{dcases}
\end{equation}
Eqs.~\eqref{eq27}--\eqref{eq:fqdothls}
rely only on the assumption
that $H'$ is in the Lindblad span at $t$
(and do not require time-independence).
However, if $H'$ is in the Lindblad span
at all times, and we have some time-independent $c_1$
and $c_2$, then integrating
Eq.~\eqref{eq:fqdothls}
 starting from $\F(t = 0) = 0$ gives
\begin{equation}
	\F(t) \le \begin{dcases}
		\frac{16 c_2^2}{c_1^2} \left(1 - 2^{-t/t_c}\right)^2 \qquad & t \leq t_c
		\\
		4 c_2 \left(\frac{c_2}{c_1^2} + (t - t_c)\right) \qquad & t \ge t_c
	\end{dcases}
	\label{eqfqt}
\end{equation}
where $t_c \coloneqq \frac{2c_2}{c_1^2} \ln 2$.
If we take $c_2 = \lVert \sum \widetilde A_j^\dagger \widetilde A_j \rVert$ (which is always $\ge \sum_j \langle \widetilde A_j^\dagger \widetilde A_j\rangle$),
then Eq.~\eqref{eqfqt} scales in the same way as Eq.~\eqref{eqlinbound}
for large $t$, but sharpens it for any finite $t$.
Additionally, Eq.~\eqref{eqfqt} shows that increasing $\F_{H'}$ (which is upper-bounded by $c_1$) does not improve
the large-$t$ scaling of $\F$, but does
decrease the time $t_c$ that it takes to attain
this scaling (we
explore some of the consequences of this in Section~\ref{sec_bw}).
These points are illustrated in the top
row of Figure~\ref{fig_fq1},
which plots $\F(t)$ from Eq.~\eqref{eqfqt}
(and the corresponding $\dot \F(t)$)
for a particular $c_2$ value and different
$c_1$ values. We also plot the initial
$\F(t) \le 4 c_1^2 t^2$ scaling
from Section~\ref{secshort}
as well as the $\F(t) \le 4 c_2 t$ bound (Eq.~\eqref{eqlinbound})
from~\cite{prx,preskill}, illustrating
how our $\F(t)$ bound in Eq.~\eqref{eqfqt} transitions from initial quadratic growth to eventual linear growth. 

It is easy to find examples for which the bound
in Eq.~\eqref{eqfqt} can be further tightened.
As observed in Section~\ref{secshort}, the optimal values
of $\beta_j,\gamma_{jk}$ are $\simeq 0$
at small enough times, and $\simeq \widetilde\beta_j,
\widetilde\gamma_{jk}$ at large enough times.
The choice
$\beta_j = \delta(t)\widetilde \beta_j$,
$\gamma_{jk} = \delta(t)\widetilde \gamma_{jk}$
considered above corresponds to interpolating between
these two extremes in the same way for \emph{all} $j,k$.
However, this is not necessarily optimal.
For example, consider a case where
$H' = \widetilde \beta_1^* L_1 + \widetilde \beta_2^* L_2
+ {\rm h.c.}$, with $|\widetilde \beta_1| \gg |\widetilde \beta_2|$,
but $\lVert \widetilde\beta_1^* L_1 + {\rm h.c.} \rVert \ll
\lVert \widetilde\beta_2^* L_2 + {\rm h.c.} \rVert$.
The optimal choice increases $\beta_2$ from
$0$ to $\widetilde\beta_2$ faster
than it increase $\beta_1$ from $0$ to $\widetilde \beta_1$.
This corresponds to the `approximate error
correction' scenario mentioned in~\cite{preskill}.
Nevertheless, for various systems of interest,
such as the damped harmonic oscillator
we analyse in Section~\ref{secqho},
the uniform $\delta(t)$ interpolation can be
optimal, and
Eq.~\eqref{eqfqt} can be tight.

In~\cite{10.1103/PhysRevResearch.2.013235}, it was shown
that the $\F(t)\sim 4 \lVert \sum_j \widetilde{A}_j^\dagger \widetilde{A}_j \rVert t$
scaling is attainable asymptotically,
by using quantum error correction techniques.
It would be interesting to investigate
under what circumstances it is possible to obtain
$\F(t)$ growth that saturates Eq.~\eqref{eqlim1}
throughout.

Finally, we note that even though we derived
Eqs.~\eqref{eq27}--\eqref{eqfqt} for the purpose
of comparing to previous results, which apply only to
the time-independent case, these equations hold
in more general settings, and we will use them
in subsequent sections. For time-dependent master equations,
Eqs.~\eqref{eq27}--\eqref{eq:fqdothls} hold at any
$t$ where $H'$ is instantaneously in the Lindblad
span. Eq.~\eqref{eqfqt} requires that $H'$ is in
the Lindblad span over the time interval $[0,t]$,
and that one has time-independent bounds $c_1$
and $c_2$. However, even for time-dependent
$c_1$ and $c_2$, one can still derive bounds
on $\F(t)$ analogous to Eq.~\eqref{eqfqt} by
integrating Eq.~\eqref{eq:fqdothls}. Furthermore,
one can integrate starting at arbitrary values of
$\F(t=0)$, not just $\F(t=0) = 0$ (as assumed
in~\cite{prx,preskill}).

\subsubsection{$H'$ not in Lindblad span (HNLS)}
\label{sec_hnls}

In the case where
$G$ cannot be set to zero in Eq.~\eqref{eq_hdecomp}
(labelled the `{Hamiltonian-not-in-Lindblad-span}' (HNLS) case by~\cite{preskill}), \cite{preskill} showed that
by error correction techniques,
is a scheme which achieves
\begin{equation}
	\F(t) = 4 t^2 \min_{\alpha,\beta_k,\gamma_{jk}}
	\lVert G \rVert^2
	\label{eqg2}
\end{equation}
for time-independent $H'$ and $L_j$
(they considered finite-dimensional systems,
so the RHS is always well-defined).
From our $\dot{\F}$ bound in Eq.~\eqref{eqlim1}, we can see
that the large-$t$ scaling of Eq.~\eqref{eqg2}
is the best possible: from Fact~\ref{fact:FVarA},
$\frac{1}{4} \F_G \le \Var(G) \le \langle G^2 \rangle
\le \| G\|^2$ for any $G$, so integrating Eq.~\eqref{eqlim0} (for a fixed, time-independent $G$)
gives $\F(t) \le 4 t^2 \|G\|^2 (1 + o(1))$.
Thus, since our methods apply to general adaptive control
schemes (under the assumption of Markovian noise), cf.~ Section~\ref{sec_control}, our results show that
the scheme of~\cite{preskill} is asymptotically
optimal (\cite{preskill} showed
that it is optimal among schemes in which error correction is applied in a particular way,
but this does not encompass all possible detection schemes).

However, we know from Section~\ref{secshort}
that at small $t$,
the QFI can in general grow faster than
Eq.~\eqref{eqg2}.
In Appendix~\ref{appC2}, we show
that the tightest bound that can be placed
on $\F(t)$,
if one uses the methods of~\cite{prx,preskill}, is
\begin{equation}
	\F(t) \le 4 t^2 \min_{\alpha \in \mathbb{R}} \lVert H' - \alpha I \rVert^2
	\label{eqfqtpreskill}
\end{equation}
(from Section~\ref{secshort}, this
bound can be saturated, to leading order in
$t$, for small $t$). This bound applies regardless of whether $H'$ is in the Lindblad span.
In the special case where Eq.~\eqref{eqg2} 
is minimised by some $G$ of the form $H' - \alpha I$, 
Eq.~\eqref{eqfqtpreskill} is tight for all $t$.
Otherwise, we
can use our results to place bounds
on $\F(t)$ that are tighter than Eq.~\eqref{eqfqtpreskill} for all $t$.

As an example, for any $H'$ (that may or may not be in the Lindblad span), we can write
$H' = G_0 + G_S$, where $G_S = \widetilde \alpha I +
\sum_j (\widetilde \beta_j^* L_j + \widetilde \beta_j L_j^\dagger)
+ \sum_{j,k} \widetilde \gamma_{jk} L_j^\dagger L_k$ for
some coefficients $\widetilde \alpha,
\widetilde{\beta_j}, \widetilde{\gamma}_{jk}$,
and choose $\alpha = \delta(t) \widetilde{\alpha}$, $\beta_j = \delta(t) \widetilde{\beta}_j$, and $\gamma_{jk} = \delta(t)\widetilde{\gamma}_{jk}$ for some real function $\delta$.
Then,
$G = H' - \delta G_S = (1-\delta)H' + \delta G_0$, so 
Eq.~\eqref{eqlim0} gives
\begin{equation}
	\dot \F \le 2 \sqrt{\F_{(1 - \delta)H' + \delta G_0} \F}
	+ 4\delta^2 \sum_j \langle \widetilde{A}_j^\dagger \widetilde{A}_j
	\rangle,
\end{equation}
where $\widetilde{A}_j \coloneqq i(\widetilde\beta_j + \sum_k \widetilde\gamma_{jk} L_k)$.
Then, using Facts~\ref{fact:FaA+bI} and~\ref{fact:FA+B} (proven in Appendix~\ref{appqfi}), we have
\begin{equation}
	\dot \F \le 2 \Big(|1-\delta|\sqrt{\F_{H'}} + |\delta| \sqrt{\F_{G_0}}\Big) \sqrt \F
	+ 4\delta^2 \sum_j \langle \widetilde{A}_j^\dagger \widetilde{A}_j
	\rangle
	\label{eq34}
\end{equation}
Suppose that $\F_{G_0}\le 4 c_0^2$, $\F_{H'} \le 4 c_1^2$,
and $\sum_j \langle \widetilde{A}_j^\dagger \widetilde{A}_j\rangle \le c_2$ for some non-negative
$c_0,c_1,c_2$. Since it is always possible to choose
$c_0 \le c_1$ (we obtain equality by taking $G_S = 0$),
and it can be checked that $c_0 > c_1$ results in looser bounds, we suppose $c_0 \le c_1$ wlog.
Then, choosing $\delta$ to minimise
the RHS of Eq.~\eqref{eq34}, we obtain
\begin{equation} \label{FdotHNLS}
	\dot \F \le \begin{dcases}
		4 c_1 \sqrt{\F}\left(1 - \frac{(c_1 - c_0)^2}{4 c_1 c_2}\sqrt{\F}\right) \quad & \sqrt{\F} \leq \frac{2 c_2}{c_1 - c_0} \\
		4 \left(c_0 \sqrt{\F} + c_2\right) \quad & \sqrt{\F} \ge \frac{2 c_2}{c_1 - c_0}
	\end{dcases}
\end{equation}
(where one should take the first case if $c_0 = c_1$).
Now, if $c_0,c_1,c_2$ are time-independent, then
integrating this bound starting
from $\F(t =0) = 0$, we arrive at\footnote{Note that
if we are in the HLS case for all $t$,
so that we can set $G_0 = 0$, and thus $c_0 = 0$,
then Eq.~\eqref{eqfhnls} reduces to Eq.~\eqref{eqfqt}. 
Indeed, Eq.~\eqref{FdotHNLS} handles the most generic time-dependent setting, where $H'$ may satisfy the HLS condition at some times but not at others.}
\begin{equation}
	\F(t) \le
	\begin{dcases}
		\frac{16 c_1^2 c_2^2}{(c_1-c_0)^4}
		\left(1 - \left(\frac{2c_1}{c_1 + c_0}\right)^{-t/t_c}\right)^2 &\quad t \le t_c \\
		y(t-t_c)^2 &\quad t \ge t_c
	\end{dcases}
	\label{eqfhnls}
\end{equation}
where $t_c \coloneqq \frac{2c_2}{(c_1 - c_0)^2}\ln\left(\frac{2 c_1}{c_1 + c_0}\right)$, and we define
the function $y$ in Appendix~\ref{appC3} (Eq.~\eqref{y}). We note that Eq.~\eqref{FdotHNLS} can be used to derive bounds on $\F(t)$ even when $c_0, c_1, c_2$ are not necessarily time-independent.
Also, per the discussion below Eq.~\eqref{eqfqt},
in some cases it will be possible to find
tighter bounds than that given by Eq.~\eqref{FdotHNLS}. 

In Appendix~\ref{appC3}, we show that
for large $t$, $y(t) \sim 2 c_0 t$,
so for $t \gg t_c$,
the bound in Eq.~\eqref{eqfhnls} scales as
\begin{equation} \label{eqFtscaling}
	\F(t) \le 4 c_0^2 t^2 (1 + \OO(t^{-1/2}))
\end{equation}
This agrees with the scaling derived below
Eq.~\eqref{eqg2}, confirming
that the scheme of~\cite{preskill} is asymptotically
optimal.
Conversely, for $t\ll t_c$,
the bound in Eq.~\eqref{eqfhnls} scales as
$\F(t) \le 4 c_1^2 t^2 + \OO(t^3)$,
matching the small-time growth
rate (see Section~\ref{secshort}).
It would be interesting to
investigate under what circumstances
Eq.~\eqref{FdotHNLS} can be saturated
at finite times, as opposed to asymptotically
(as achieved by the schemes in~\cite{preskill}).

\subsection{Damped harmonic oscillator}
\label{secqho}

\begin{figure*}[t]
\includegraphics[width=0.4\textwidth]{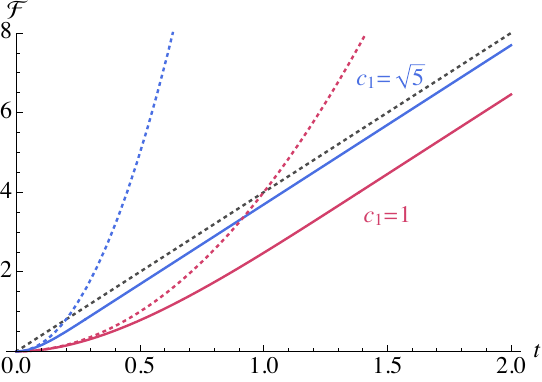}
\qquad
\includegraphics[width=0.4\textwidth]{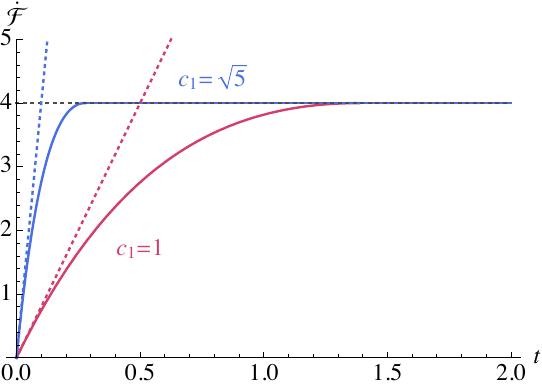}

\vspace{0.3cm}

\includegraphics[width=0.4\textwidth]{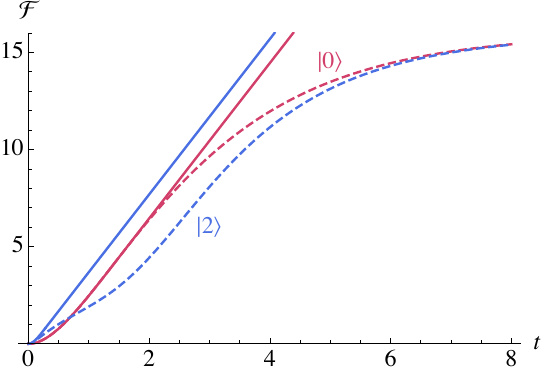}
\qquad
\includegraphics[width=0.4\textwidth]{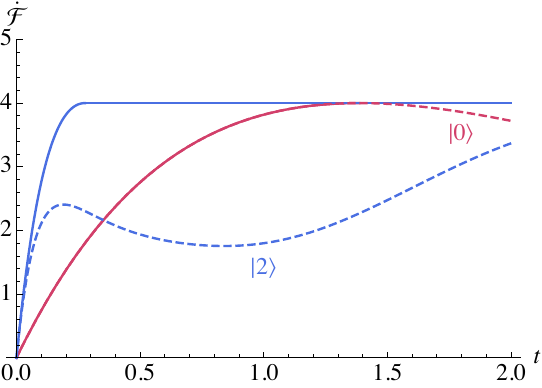}
	\caption{\emph{Top-left panel:}
	Plots of $\F(t)$ bounds in the HLS case from Section~\ref{secqfi}, compared to those from previous work.
	The solid curves correspond to the bound
	on $\F(t)$ from Eq.~\eqref{eqfqt};
	the red curves take $c_1 = 1$,
	$c_2 = 1$ (all quantities are taken to be dimensionless),
	and the blue curves take
	$c_1 = \sqrt 5$, $c_2 = 1$.
	The gray dotted line shows
	the $\F(t) \le 4 c_2 t$ bound from~\cite{prx,preskill} (Eq.~\eqref{eqlinbound}),
	while the red and blue dotted curves
	correspond to the quadratic bounds
	from Eq.~\eqref{eqfqtpreskill},
	for the appropriate $c_1$.
	\emph{Top-right panel:} As per the top-left
	panel, but plotting $\dot \F(t)$ bounds.
	This illustrates how the large-$t$ scaling
	of our bounds is only affected by $c_2$,
	and not $c_1$.
	\emph{Bottom-left panel:}
	Solid curves as in top row plots; dashed curves plot $\F(t)$ for different detection schemes
	with a damped harmonic
	oscillator probe, as considered
	in Section~\ref{secqho} (taking $\gamma = 1$
	and $n_T = 0$).
	The dashed red curve corresponds
	to preparing an oscillator in its ground
	state, and allowing it to evolve unimpeded
	under the influence of the linear
	forcing. This attains our $\F(t)$ bound (solid red curve)
	for $c_1 = c_2 = 1$ at $t \le t_c = 2(c_2/c_1^2) \ln 2
	\simeq 1.4$,
	but then asymptotes to a constant value
	of $\F(t)$.
	The blue dashed curve corresponds to preparing an
	oscillator in the $N=2$ Fock state, and allowing
	it to evolve. This has faster
	$\F(t)$ growth rate at small times, but that growth
	slows down for $t \gtrsim t_c = \frac{2\ln2}{5} \simeq 0.3$.
	As well as representing
	the Eq.~\eqref{eqfqt} bound with $c_1 = 1$,
	the solid red curve also corresponds to $\F(t)$ for a harmonic oscillator
	prepared in its ground state,
	but coupled to a continuum of ancillae
	for $t \ge t_c = 2 \ln 2$, to maintain
	the system at $\dot\F = 4 c_2$.
	\emph{Bottom-right panel:}
	As per the bottom-left panel,
	but plotting $\dot \F(t)$.
	}
	\label{fig_fq1}
\end{figure*}

In this subsection, we study
the case where the probe system is a (weakly)
damped harmonic oscillator, subject to some near-resonant
forcing.
This is a rather more complex and interesting
example than the single spin from Eq.~\eqref{eq_tlsdephase},
and illustrates different points of theoretical
and practical interest.
Since an oscillator is an infinite-dimensional system,
the methods of~\cite{prx,preskill} do not apply
directly. While one can sometimes consider restrictions to a finite-dimensional
subspace---for instance, to states with occupation
number $\le N_0$ for some $N_0$---physical processes
often generate evolution that cannot be
contained within any finite-dimensional subspace
(e.g., any continuous range of coherent state
amplitudes).
In addition, our methods allow us to treat time-dependent
$H'$, as we will further expand on in Section~\ref{sec_bw}.
From an experimental perspective, many
systems of interest, such as microwave, optical,
or acoustic modes, are well-described by oscillators.

We will show that simple, physically relevant
forms for $H$ are such that $H'$ is in the Lindblad
span (cf.~Section~\ref{sec_hls}), so we can obtain 
an $\F$-independent bound on $\dot{\F}$ by choosing $G = 0$ in Eq.~\eqref{eqlim0}.
Moreover, this bound can be attained, for
on-resonance forcings, by an oscillator
prepared in any coherent state. However, non-classical states
can lead to faster short-time QFI growth,
as we will illustrate.

In the rotating wave approximation (RWA), the
master equation for a harmonic oscillator interacting
with a Markovian environment at temperature
$T$ is~\cite{10.1007/978-3-540-28574-8}
\begin{align}
	\dot \rho &= - i [H,\rho] + \gamma(n_T+1)
	\left(a \rho a^\dagger - \frac{1}{2}
	\{a^\dagger a, \rho \}\right) \nonumber \\
	&\quad + \gamma n_T \left(a^\dagger \rho a -
 \frac{1}{2}
	\{a a^\dagger, \rho \}\right).
	\label{eq_oscsme}
\end{align}
Here, $a^\dagger$ and $a$ are the oscillator's
creation and annihilation operators,
$H$ is the RWA interaction Hamiltonian,
$\gamma$ is the coupling to the Markovian
bath (if the oscillator has no other influences
acting on it, then $\gamma$ is its damping rate),
and $n_T$ is the thermal occupation number
corresponding to the temperature of the bath
($n_T = (e^{\omega/T} + 1)^{-1}$,
where $\omega$ is the resonant frequency of
the oscillator).
Thus, we have two Lindblad operators,
\begin{equation}
	L_1 = \sqrt{\gamma(n_T + 1)} a, \quad
L_2 = \sqrt{\gamma n_T} a^\dagger.
\end{equation}

The simplest kind of forcing we can consider
is linear in the creation/annihilation operators and in $g$, so that
\begin{equation} H' = i \epsilon(t) a^\dagger + {\rm h.c.}, \end{equation}
corresponding to a near-resonant force on the oscillator
(with $\epsilon(t)$ representing the fiducial amplitude
and phase of the forcing, and $g$ representing its relative amplitude).
In this case, $H'$ can
be decomposed in the Lindblad span as in Eq.~\eqref{eqhls}, but the decomposition is not unique.
To obtain the best large-$\F$ bound
from Eq.~\eqref{eqlim0}, we want to
minimise $\sum_j \widetilde{A}_j^\dagger \widetilde{A}_j$, which is
always a multiple of the identity in this case, since $\widetilde{\gamma}_{jk} = 0$ for any such decomposition. 
To do so,
we should choose
$\widetilde{\beta}_1 = \frac{i\epsilon}{\sqrt\gamma}\frac{\sqrt{n_T+1}}{2 n_T + 1}$ and $\widetilde{\beta}_2 = \frac{\sqrt{n_T}}{2 n_T + 1} \frac{-i \epsilon^*}{\sqrt\gamma}$, giving
$\sum_j \widetilde{A}_j^\dagger \widetilde{A}_j =  \frac{|\epsilon|^2}{\gamma(2n_T+1)} I$.\footnote{This form is
not surprising, since one expects $\dot \F$ to increase
more slowly in the presence of thermal noise.
The position fluctuations of an oscillator in a thermal
state are $\langle x^2 \rangle \propto 2 n_T + 1$,
and $\sum_j \widetilde{A}_j^\dagger \widetilde{A}_j$
scales with $n_T$ in the inverse manner.}
Note that if any of $\epsilon$, $n_T$,
or $\gamma$ are time-dependent, then this quantity will
be time-dependent.

On the other hand, the short-time QFI
growth rate is
determined by $\F_{H'}$ (cf.~Section~\ref{secshort}).
From Fact~\ref{fact:FVarA},
$\F_{H'} \le 4 \Var(H')$,
with equality for pure states.
This means that $\F_{H'}$
is bounded by the fluctuations in the appropriate
quadrature (for $\epsilon$ real, simply the momentum
fluctuations of the oscillator).
Since a coherent state is a minimum-uncertainty
state, with equal uncertainties in all quadratures,
these fluctuations can only be made large
by putting the oscillator into a `non-classical'
state---that is, a state which is not a probability
mixture of coherent states.

In the notation of Eq.~\eqref{eq:fqdothls}, we have
$c_2 =  \frac{|\epsilon|^2}{\gamma(2n_T+1)}$
and $c_1^2 \le 2 |\epsilon|^2 \sigma_{\widetilde p}^2$,
where $\widetilde p \coloneqq \frac{1}{\sqrt{2}|\epsilon|} H'$ 
is the appropriate quadrature operator, and
$\sigma_{\widetilde p}^2 \coloneqq
\langle \widetilde p^2 \rangle  - \langle \widetilde p \rangle^2$ is its variance.
If $c_1$ and $c_2$ are time-independent,
then starting from $\F(t = 0) = 0$,
Eq.~\eqref{eqfqt} gives
\begin{equation}
	\F(t) \le \begin{dcases}
		\frac{8 |\epsilon|^2}{\gamma_T^2 \sigma_{\widetilde p}^2}
		\left(1 - e^{-\gamma_T t  \sigma_{\widetilde p}^2}\right)^2 &\quad
		t \leq \frac{\ln 2}{\gamma_T \sigma_{\widetilde p}^2} \\
		\frac{4 |\epsilon|^2}{\gamma_T}
		t
		- \frac{2  (\ln 4 - 1)|\epsilon|^2}{\gamma_T^2 \sigma_{\widetilde p}^2}
		&\quad
		t \ge \frac{\ln 2}{\gamma_T \sigma_{\widetilde p}^2}
	\end{dcases}
	\label{eq_oscbound1}
\end{equation}
where $\gamma_T \coloneqq \gamma (2n_T + 1)$. 
If $c_1$ and $c_2$ vary with time, then we
could integrate Eq.~\eqref{eq:fqdothls} to
obtain analogous bounds.

The simplest situation for which we can analyse
the actual behaviour of $\F(t)$ (as opposed to an upper bound)
is when $H = gH'$ and $n_T = 0$. In
this case, coherent states evolve to other coherent states, so
if the oscillator
starts in a coherent state, its state remains pure
and coherent throughout~\cite{10.1007/978-3-540-28574-8}.
The coherent state $|\alpha(t,g)\rangle$,
with amplitude $\alpha(t,g)$,
evolves as $\dot \alpha = g \epsilon - \frac{\gamma}{2}
\alpha$.
Consequently, if
$\epsilon$ and $\gamma$ are time-independent, then
\begin{align}
	\alpha(t,g) &=\frac{2 g \epsilon}{\gamma}(1 -
	e^{-\gamma t /2}) + e^{-\gamma t/2}\alpha(t=0,g) \nonumber \\ &\eqqcolon g \alpha_\epsilon(t) + e^{-\gamma t/2}\alpha_0
\end{align}
To find $\F(t)$, we can use the fact that if
$\rho(t,g)$ is pure for all $g$ (as is true
here), then $\L = 2\rho'$. Consequently, writing $\rho = |\psi\rangle\langle\psi|$,
we have $\F = 4| \|\varphi_\perp\rangle\|^2$,
where $|\varphi_\perp\rangle$ is the
component of $\partial_g\ket{\psi}$
orthogonal to
$\ket\psi$. 
To evaluate $|\varphi_\perp\rangle$,
we note that $D(\beta)D(\delta) =
e^{(\beta \delta^* - \beta^* \delta)/2}D(\beta+\delta)$
for any $\beta,\delta \in \mathbb{C}$, where $D$ is the
displacement operator which generates coherent
states, $D(\beta) = e^{\beta a^\dagger - \beta^* a}$.
That is, $D(\beta)D(\delta)$ and $D(\beta + \delta)$
are equal, up to phase. 
Furthermore,
\begin{equation}
	\partial_g D(g \alpha_\epsilon)
	= \partial_g e^{g (\alpha_\epsilon a^\dagger
	- \alpha_\epsilon^* a)}
	= (\alpha_\epsilon a^\dagger - \alpha_\epsilon^* a)
	D(g \alpha_\epsilon)
\end{equation}
Consequently, if $\alpha_0$ is independent of $g$, then
\begin{align}
	\ket{\varphi _\perp}
	= e^{i \phi}
	D(e^{-\gamma t/2}\alpha_0)
	(\alpha_\epsilon a^\dagger - \alpha_\epsilon^* a)
	D(g \alpha_\epsilon) \ket 0
\end{align}
for some $\phi \in \mathbb{R}$.
Evaluating $\|\ket{\varphi_\perp}\|^2$, we obtain
\begin{equation}
	\F(t) = 4 |\alpha_\epsilon|^2 = \frac{16|\epsilon|^2}{\gamma^2} (1 - e^{-\gamma t/2})^2
	\label{eq_ftosc}
\end{equation}
for any $\alpha_0$.
Comparing this to the $t \leq t_c = \ln2 / (\gamma \sigma_{\widetilde{p}}^2)$ bound in
Eq.~\eqref{eq_oscbound1}, we see that they are
equal for $\sigma_{\widetilde p}^2 = \frac{1}{2}$,
which is true for any coherent state. Hence,
preparing the
oscillator in a coherent state and letting
it evolve unimpeded saturates
the corresponding $\F(t)$ bound for $t \leq t_c$.
At $t = t_c$, Eq.~\eqref{eq_ftosc}
attains the $\dot \F = 4c_2 = 4 |\epsilon|^2/\gamma$
bound from Eq.~\eqref{eq:fqdothls}.

For an unperturbed
oscillator, $\dot \F(t)$
starts to decrease for $t > t_c$; left to evolve
by itself, the oscillator's state
would simply asymptote to $\alpha = 2 g \epsilon / \gamma$,
so $\F$ would asymptote to $16 |\epsilon|^2/\gamma^2$,
and $\dot \F$ would tend towards zero.
This is illustrated by the dashed red curves
in the bottom row of Figure~\ref{fig_fq1}.
To avoid this, we need to transfer
some of the information in the state of the oscillator
into ancillary degrees of freedom.
The simplest way to achieve this is to couple the
oscillator to a continuum of bosonic modes,
such that energy loss into the modes acts,
from the oscillator's point of view, like
an extra source of damping.
For example, if we viewed the oscillator
as a cavity mode,
we could couple this mode to a waveguide.
To maintain $\dot \F$ at the optimal value,
it turns out that the extra damping
rate should be the same as the Markovian
damping rate $\gamma$, corresponding to critical coupling
of the oscillator to the continuum, giving an
effective damping rate of $2\gamma$.

If one switches on this continuum coupling at $t_c$,
then it maintains the system at the $\dot \F = 4 |\epsilon|^2 / \gamma$ limit, and so saturates the $\F(t)$ bound
in Eq.~\eqref{eq_oscbound1} throughout.
This means that we can attain the bound in
Eq.~\eqref{eq_oscbound1} (with $\sigma_{\widetilde{p}}^2 = 1/2$),
which corresponds to the solid red curves
in Figure~\ref{fig_fq1}.
Practically, of course, it may be simpler to
couple the mode to the continuum from the beginning;
in that case, $\dot \F$ asymptotes towards
the $4|\epsilon|^2/\gamma$ limit over a few
damping times.

If we allow the oscillator to be in
a non-classical state, then the state-independent
$\dot \F \le {4 |\epsilon|^2}/{\gamma_T}$
bound still applies, but $\dot \F$ can increase
faster starting from zero.
From Section~\ref{secshort}, the short-time
growth rate of $\F$ is set by $\F_{H'}(t=0)$,
and $\F_{H'} \le \Var(H')$ with equality
for pure states, so we obtain faster growth
by using states with larger $H'$ fluctuations.
This is illustrated by the dashed blue curves
in the bottom panels of Figure~\ref{fig_fq1},
which correspond to preparing the oscillator
in an $N=2$ Fock state.

Since large fluctuations require large energies,
one question we can ask is how large $\F_{H'}$
can be, given some bound on $\langle N \rangle$.
Writing $\widetilde x \coloneqq \frac{1}{\sqrt{2}|\epsilon|}(\epsilon a^\dagger
+ {\rm h.c.})$
for the quadrature orthogonal to $\widetilde{p}$,
we have the uncertainty principle relation
$\sigma_{\widetilde x}^2 \sigma_{\widetilde p}^2 \ge 1/4$,
as well as the relation ${\widetilde x}^2 + {\widetilde p}^2 = 2 N + 1$, where $N$ is the oscillator's number operator. 
Together, these
imply that~\cite{10.1103/PhysRevLett.111.173601}
\begin{equation}
	\sigma_{\widetilde p}^2 \le \frac{1}{2} + \langle N\rangle
+ \sqrt{\langle N \rangle (\langle N \rangle + 1)}
	\label{eqsigmax}
\end{equation}
This inequality is saturated iff the state
is a minimum-uncertainty state with zero mean,
i.e., a squeezed coherent state. (For
comparison, the $N$th Fock state
has $\sigma_{\widetilde p}^2 = N + \frac{1}{2}$.)
Hence,
since $\Var(H')
	= 2 |\epsilon|^2 \sigma_{\widetilde p}^2$, the maximum $\F(t)$ at small
$t$, for given $\langle N \rangle$, is
attained by preparing the oscillator
in a squeezed coherent state, with squeezing quadrature appropriate
to the phase of the forcing
(this is analogous to the conclusions in~\cite{10.1103/PhysRevLett.111.173601}). Since enhanced $\widetilde p$ fluctuations correspond
to suppressed
$\widetilde x$ fluctuations, and the forcing
displaces the state
in the $\widetilde x$ direction,
this is intuitive from a quantum fluctuations viewpoint.

The bound in Eq.~\eqref{eqsigmax} illustrates
that given constraints on our state, we can often
obtain sensible results even for infinite-dimensional
systems, without having to restrict to a finite-dimensional
subspace (as done in e.g.,~\cite{preskill}).
Another example of this is the case of quadratic forcing,
$H' = \omega_f N$ (as can arise in e.g., optomechanical
systems). Taking $n_T = 0$ for simplicity, we have
$\widetilde\gamma_{11} = \omega_f/\gamma$ (in the notation of Eq.~\eqref{eqhls}), so
$\langle \widetilde A_1^\dagger \widetilde A_1 \rangle = {\omega_f^2}
\langle N \rangle/\gamma$. So, if we restrict
to $\langle N \rangle \le \overline N$, then
we can take $c_2 = {\omega_f^2}\overline N/\gamma$
and $c_1^2 \le 2 |\epsilon|^2 (\frac{1}{2} + \overline N +
\sqrt{\overline N (\overline N + 1)})$, and
obtain the corresponding bounds using Eq.~\eqref{eqfqt}.

The linear ($H' = i \epsilon a^\dagger + {\rm h.c.}$)
and quadratic ($H' \propto N$) forcings
are simple, physically important examples which 
are in the Lindblad span. For $H'$ other than
linear combinations of these, we are in the HNLS
case, so $\dot \F$ can in principle grow
without bound. Examples include the `Kerr effect'
Hamiltonian $H' \propto N^2$ considered in~\cite{preskill},
or (for $n_T = 0$) quadratic Hamiltonians
of the form $H' = f a^2 + {\rm h.c.}$
The latter can arise
from coupling to a signal oscillating at close
to twice the natural frequency of the oscillator.
Since this kind of forcing preserves
the occupation number parity of the oscillator,
whereas Lindblad jumps change it, 
schemes such as those based on `cat codes'~\cite{10.1103/PhysRevA.59.2631}
could be used to attain large $\dot \F$.

While the results in this subsection
apply regardless of the $g$-independent dynamics
of the system, they do assume
that the oscillator is described by the master equation
in Eq.~\eqref{eq_oscsme}. In particular,
they rely on the assumption
that the noise is Markovian.
If the correlation time
of the system-environment interaction is
not small enough to be neglected, then techniques
such as dynamic decoupling~\cite{10.1103/PhysRevA.58.2733,10.1103/PhysRevLett.82.2417}
can violate the above bounds.

\section{Sensitivity bandwidth}
\label{sec_bw}

In Section~\ref{secqfi}, we considered
a Hamiltonian parameterised
by a single unknown scalar parameter $g$.
Apart from the value of $g$, we assumed
that the form---in particular, the time
dependence---of our Hamiltonian
was known \emph{a priori}.
However, as discussed in the introduction,
in many circumstances we are interested in a range of possible
time dependences.

There are a number of ways to formalise
this more general problem. We may view
our task as simultaneously estimating 
a large number of parameters controlling
the Hamiltonian, as in the `waveform estimation'
problem discussed in~\cite{10.1103/PhysRevLett.106.090401}.
Alternatively, we could attempt to estimate
a single parameter, e.g., one controlling the 
overall strength of our signal,
with other parameters
viewed as `nuisance parameters'~\cite{10.1088/1751-8121/ab8ef3}.
The latter is the appropriate approach in e.g., axion
dark matter detection, where we are interested in detecting
the presence of a signal,
and not in its detailed time dependence
(at least for initial discovery purposes).

As a simple example, we will study the scenario
of a scalar time dependence;
that is, where the Hamiltonian has the form
$H(t) = f(t,[g_i]) \widetilde{H} + H_c$, 
where $[g_i]$ is some vector of scalar parameters, 
$f(t,[g_i])$ is a scalar function, 
$\widetilde{H}$ is a $t$- and $[g_i]$-independent operator, and $H_c$ is a $[g_i]$-independent operator.
Examples of this form include
oscillatory signals with known
coupling type, but \emph{a priori}
unknown time dependence, as arise in searches
for axion dark matter or gravitational waves.

\subsection{Waveform estimation}
\label{sec_waveform}

We can consider a toy version of the waveform
estimation problem by taking our signal
to be stepwise constant. Specifically, suppose that the
parameter $g_i$ controls the signal
during the time
interval $[t_i,t_{i+1})$,
so that $H = H_c + \sum_i g_i \mathbf{1}_{[t_i,t_{i+1})}
\widetilde H$. If we take the $g_i$ parameters
to be independent, then we can simply consider independent
single-parameter
estimation problems on each time interval $[t_i, t_{i+1})$;
for each of these problems, the bounds we derived in Section~\ref{secqfi}
apply to the QFI $\F_i$ for each parameter $g_i$.

Now, suppose that we have some time-independent bound $\dot \F_{\rm max}$
on the QFI growth rate. As an example, if $\widetilde{H}$ is in the Lindblad span, then Eq.~\eqref{eq:fqdothls} gives such a bound if $c_2$ is time-independent. 
Then, over
the total time $t_{\rm tot} \coloneqq \sum_i (t_{i+1} - t_i)$,
the sum of the individual QFIs $\F_i$ is bounded by
$\sum_i \F_i \le \dot \F_{\rm max} t_{\rm tot}$.
The sum of the QFIs gives an upper bound
for e.g.,
how well we can distinguish between two specific hypotheses
for the vector $[g_i]$.
 Furthermore, suppose that we have a time-independent bound $\F_{\widetilde{H},\max}$ on $\F_{\widetilde{H}}$. 
 Then, by setting $G=H'$ in Eq.~\eqref{eqlim0}
 and integrating from $\F_i = 0$ at $t = t_i$, we
 have $\F_i(t) \le \F_{\wH,\rm max} (t - t_i)^2$
 for $t \ge t_i$.
Substituting this back into Eq.~\eqref{eqlim0},
we have $\dot \F_i \le 2 \F_{\wH,\rm max} (t - t_i)$.
Thus, the time taken to attain $\dot \F_i = \dot\F_{\rm max}$
is at least $t_c \sim \dot \F_{\rm max} / \F_{\wH,\max}$.\footnote{Note that even though $H'$ in each interval 
is time-independent, the $\F \le 4 c_2 t$ bound from
\cite{prx,preskill} is not useful; we need the
finite-time bounds derived in this paper to see
that $\dot\F = 4 c_2$ cannot be attained immediately.}
We can make the constant factor more precise
using the results from
Sections~\ref{sec_hls} and~\ref{sec_hnls};
as discussed below Eq.~\eqref{eqfqt}, in some circumstances
it will also be possible to parametrically tighten this bound.

Hence, if $|t_{i+1} - t_i| \lesssim t_c$ for every $i$,
then $\sum_i \F_i \simeq \dot \F_{\rm max} t_{\rm tot}$
cannot be attained. 
Conversely, for the damped harmonic oscillator analysed in
Section~\ref{secqho}, if we use the state-independent
bound $\dot \F_{\rm max} = 4 |\epsilon|^2 / \gamma$, then
Eqs.~\eqref{eq_oscbound1} and
\eqref{eq_ftosc}
show that if $|t_{i+1} - t_i| \gtrsim t_c$,
then we can approximately attain
$\sum \F_i \simeq \dot \F_{\rm max} t_{\rm tot}$. 
Viewing the time-step widths as corresponding
to the inverse bandwidth of our
signal, the bandwidth $\Delta \omega$ over which we can
attain $\sum \F_i \simeq \dot \F_{\rm max} t_{\rm tot}$ is
$\Delta \omega \lesssim {\F_{\widetilde H,\rm max}}/{\dot \F_{\max}}$. 

This step-function analysis is obviously rather
crude; we leave a more sophisticated analysis (along the lines
of~\cite{10.1103/PhysRevLett.106.090401}, which
treats the noiseless case) to future work.
However, it does illustrate the important point that even though the peak sensitivity to any given
signal is bounded by $\dot \F_{\rm max}$,
the sensitivity bandwidth is controlled by the finite-time
QFI growth rate, and is bounded by
$\F_{\widetilde H,\rm max}/\dot \F_{\rm max}$. In the harmonic oscillator example, we can take $\dot \F_{\rm max} = 4 |\epsilon|^2/\gamma $,
which is independent of the system's state, and
can be attained by an oscillator in its ground
state---in that sense, there is no benefit to using
`non-classical' states. However, as shown in Section~\ref{secqho}, such states
\emph{can} enhance the short-time growth rate
of $\F$, and consequently, the range
of different time dependences we are sensitive to.

\subsection{Nuisance parameters}
\label{sec_nuisance}

\begin{figure*}[t]
\includegraphics[width=0.47\textwidth]{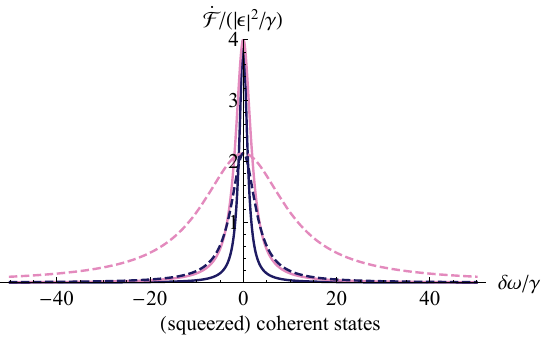} \enspace
\includegraphics[width=0.47\textwidth]{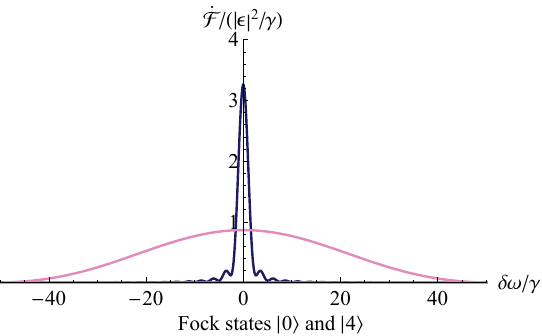}
	\caption{Plots of $\dot \F$
	(at large times)
	for a damped harmonic oscillator, as a function
	of the detuning $\delta \omega$ of the signal
	frequency from the oscillator's resonant 
	frequency. The oscillator is assumed to be governed
	by the master equation in Eq.~\eqref{eq_oscsme}
	(with $n_T = 0$).
	\emph{Left panel}: The blue curve
	corresponds to an oscillator that is critically
	coupled to ancillary degrees of freedom (e.g., a cavity
	mode coupled to a waveguide output),
	and initialised in its ground state.
	As discussed in Section~\ref{secqho},
	this setup attains the $\dot \F 
	\le 4 |\epsilon|^2 / \gamma$ bound
	at large enough times, for a resonant forcing.
	The pink curve corresponds to an oscillator
	critically coupled to a source
	of squeezed vacuum, with squeezing
	parameter $G_s = 4$.
	The dashed pink curve corresponds
	to the oscillator being overcoupled
	to the same squeezed state source by
	a factor $\sim 4 G_s$, which decreases
	on-resonance sensitivity by a constant
	factor, but results
	in a sensitivity bandwidth $\sim G_s$ times larger
	than for preparation in the ground state (Section~\ref{sec_nuisance}).
	This is illustrated by the dashed blue curve,
	which corresponds to preparation in the ground
	state, overcoupled such that the on-resonance
	sensitivity matches that of the dashed pink curve.
	\emph{Right panel:} Plot of the time-averaged
	rate of accumulation of (classical) Fisher information
	for a scheme in which the oscillator is prepared in a Fock
	state, and then measured in the Fock basis after
	evolving for a time $t_1$, as
	described in Section~\ref{sec_prepare}.
The blue curve corresponds to preparation
	in the ground state, followed by measurement at
	the time which optimises on-resonance sensitivity
	($t_1 \simeq 2.5/\gamma$), while the pink
	curve corresponds to preparation
	in the $N=4$ state, followed by
	measurement at the time which optimises on-resonance
	sensitivity ($t_1 \simeq 0.12 /\gamma$).
	}
    \label{fig_bw1}
\end{figure*}

Estimating all of the parameters affecting
a signal,
as in waveform estimation, is at least
as hard as estimating some parameters with others
viewed as `nuisance parameters.' Consequently,
we expect similar conclusions about sensitivity
bandwidth to apply to the nuisance
parameter case. However, it is still
useful to see explicitly how estimating
a single parameter, in the large-$\F$ regime,
is affected by nuisance parameters which control
the time dependence of $H'$, as we will analyse in 
this subsection. In particular, we will illustrate
that, even after long times,
the range of different signals for which
near-peak sensitivity can be maintained is
controlled by the short-time QFI growth rate.

Specifically, suppose that $H = H(t,g,h)$ depends on an additional
parameter $h$.
For a given value of $h$, the distinguishability of different
$g$ values is set by the QFI with respect to $g$.
We will investigate the range of different $h$
values over which the QFI growth rate can be close-to-optimal.
This sets an upper bound on the range of $h$ for which 
a given measurement scheme can have near-optimal sensitivity
to $g$, if the measurement scheme is chosen without knowing $h$.

For simplicity, we will suppose that at $g=0$, we have
$\partial_h H = 0$ (i.e., that
$g=0$ corresponds to `no forcing'). Henceforth, we will write $\partial \equiv \partial_h$, and equalities holding at $g=0$
will be notated via $\overset{0}{=}$,
e.g., $\partial H \eqg 0$.\footnote{More
generally, we could set $\partial H = 0$
at $g = g_0$ for any $g_0$; the choice
of $g=0$ is arbitrary but simplifies the presentation.}
As an example, if $g$ corresponds to the overall strength
of the forcing, and $h$ controls its precise
form, then at $g=0$ the $h$ parameter should have no effect. 
We will assume that $\rho$ is $g$- and $h$- independent at
some initial time $t_0$.
Taking the derivative of the master equation (Eq.~\eqref{eqsme}) with respect to $g$,
we have
\begin{equation}
	\dot \rho' = - i [H',\rho]
	- i [H,\rho'] + \sum_j
	\left(L_j \rho' L_j^\dagger - \frac{1}{2}\{
		L_j^\dagger L_j,\rho'\}\right)
	\label{eqsmeprime}
\end{equation}
and then taking the derivative of this with
respect to $h$,
\begin{equation}
	\partial \dot \rho' \eqg - i [\partial H',
	\rho] - i[H,\partial \rho'] +
	\sum_j\left(L_j \partial\rho' L_j^\dagger - \frac{1}{2}\{
		L_j^\dagger L_j,\partial\rho'\}\right)
\end{equation}
since $\partial H \eqg 0$ and consequently $\partial \rho \eqg 0$. 
This has the same form as Eq.~\eqref{eqsmeprime}, which suggests that we can use the bounds derived in Section~\ref{secqfi}, replacing $H'$ with $\partial H'$ and $\rho'$ with $\partial \rho'$.
More precisely, consider the state
$\sigma(t,g,h)$, defined via the master equation
\begin{equation}
	\dot \sigma = - i [H^{(\sigma)},\sigma]
	+ \sum_j
	\left(L_j \sigma L_j^\dagger - \frac{1}{2}\{
		L_j^\dagger L_j,\sigma\}\right)
\end{equation}
where $H^{(\sigma)}(t,g,h) \coloneqq
H(t,g=0,h) + \partial H(t,g,h)$,
along with the initial condition 
\begin{equation} \sigma(t=t_0) = \rho(t=t_0)
\end{equation} 
for all $g$ and $h$. Noting that $H^{(\sigma)} \eqg H$, we see that $\sigma$ obeys the same master equation as $\rho$ at $g = 0$, implying that $\sigma \eqg \rho$ for all $t \geq t_0$. Moreover, we have
\begin{equation}
	\dot \sigma' \eqg - i [\partial H',
	\sigma] - i[H,\sigma'] +
	\sum_j\left(L_j \sigma' L_j^\dagger - \frac{1}{2}\{
		L_j^\dagger L_j,\sigma'\}\right)
\end{equation}
so $\sigma'$ obeys the same differential equation as $\partial\rho'$ for $g = 0$. From the initial condition, we have $\sigma'(t=t_0) = \partial\rho'(t=t_0) = 0$, so it follows that $\sigma' \eqg \partial\rho'$ for all $t \geq t_0$.
Therefore,
$\partial \L \eqg \L^{(\sigma)}$
where $\L$ is the SLD for $\rho$ and $\L^{(\sigma)}$ is the SLD for $\sigma$ (cf.~Eq.~\eqref{eq_ldefn}).
These identities will allow us to use the bounds in
Section~\ref{secqfi}, applied
to $\F^{(\sigma)} \coloneqq \tr (\sigma (\L^{(\sigma)})^2)$,
to make statements about $\F = \tr (\rho \L^2)$.

We now specialise to a scalar time dependence
of the form $H = g (f_0(t) + h f_1 (t)) \widetilde H + H_c$,
where $\widetilde{H}$ is independent of $t$, $g$, and $h$,
and $H_c$ is independent of $g$ and $h$.
Then, the basic quantities are linear in $h$, so we can write
$H' = H'_0 + h \partial H'$,
$\rho' = \rho'_0 + h \partial \rho'$ (from Eq.~\eqref{eqsmeprime}),
and $\L = \L_0 + h \partial \L$ (since $\L$ is linear in $\rho'$), where zero subscripts indicate the values at $h =0$.
Hence, from Eq.~\eqref{eqfqdot}, 
$\dot \F$ is quadratic in $h$, with the $h^2$ term
given by
\begin{equation}
	\partial^2 \dot \F \overset{0}{=} 2 i \tr
	(\rho [\partial H',\partial \L])
	- \sum_j \tr(\rho [\partial \L,L_j]^\dagger
	[\partial \L,L_j]).
\end{equation}
Since $\partial \L \eqg \L^{(\sigma)}$
and $\partial H' = (H^{(\sigma)})'$, 
we have $\partial^2 \dot \F \eqg \dot \F^{(\sigma)}$ by Eq.~\eqref{eqfqdot}.

Supposing that $\wH$ is in the Lindblad span,
and we can bound $\dot \F$ 
as in Eq.~\eqref{eq:fqdothls},
then the large-$\F$ bound
will scale
with $h$ as 
$\dot \F_{\rm max} \propto (f_0 + h f_1)^2$
(for any $t$).
Consequently, we will suppose
that $\dot \F \le (f_0 + h f_1)^2 \dot \F_{b}$
for some time-independent $\dot \F_b$.
We are interested in quantifying the range
of different time dependences for
which we can (approximately) attain
this bound, using the same detection system.
To do this, we can assume that $\dot \F$
attains the bound for $h=0$,
and ask what is required for $\dot \F$ at
positive and negative $h$ values to also
(almost) attain the bound.
Equating the $h^2$ terms in $\dot \F$
and $\dot \F_{\rm max}$, we must have
	$\partial^2 \dot \F \simeq f_1^2 \dot \F_b$.

Since $\partial^2 \dot \F \eqg \dot \F^{(\sigma)}$, 
then if $f_1(t) = 0$ for $t$ less than some time
$t_f$, we can use the finite-time QFI
bounds derived in Section~\ref{secqfi}.
Taking the simplest example of a step function
time dependence, i.e., $f_1 = \Theta(t - t_f)$,
we have from Eqs.~\eqref{eq:fqdothls} and \eqref{eqfqt}
that
the time taken to attain 
$\dot \F^{(\sigma)} \simeq f_1^2 \dot \F_b$
at $g = 0$ is at least
$t_c \sim {\dot \F_b}/{\F_{\widetilde H,\rm max}}$,
where $\F_{\wH,\rm max}$ is a bound on $\F_{\wH}$.

From the discussion below Eq.~\eqref{eqlim0},
to saturate the large-$\F$ bound $\dot \F_{\rm max}$,
we need
$[L_j,\L] \sqrt{\rho} = 2 \widetilde A_j \sqrt{\rho}$
for all $j$. In our case,
since $\widetilde A_j \propto (f_0 + h f_1)$,
we need different values of $\L$ for different
$h$ values, if $f_1 \neq 0$. 
For the scenario where $f_1 = 0$ for $t < t_f$,
but changes suddenly at $t_f$,
$\dot \F_{\rm max}$
will also change suddenly (for non-zero $h$)
but $\L$ will take some time to change,
so the system will take some time to attain
the new $\dot \F_{\rm max}$. This is what
the analysis above quantifies.

As expected, $t_c$ is the same parametric timescale
we found for waveform estimation.
However, the interpretation is slightly different. For different values of $h$, the evolution of the state will be different, and in particular, the optimal measurement (with classical
Fisher information equal to the QFI) 
may be different. Consequently, 
attaining $\dot \F$ close to the bound
is necessary, but potentially not sufficient, for there to be a single measurement scheme with close-to-optimal
sensitivity for a range of $h$ values.
However, as demonstrated by e.g., 
the explicit prepare-measure-reset
scheme discussed below (Section~\ref{sec_prepare}),
it is generally possible to obtain
a sensitivity bandwidth scaling
as $t_c^{-1}$, though the numerical prefactors may differ.

For simple systems, we can explicitly
analyse both the QFI growth and the
performance of specific measurement
schemes. 
An example is the detection of a
near-resonant force acting
on a damped harmonic oscillator.
As we found in Section~\ref{secqho},
the bound $\dot \F \le 4 |\epsilon|^2 / \gamma$
can be attained
by an oscillator in a coherent state, for 
an on-resonance forcing.
However, as illustrated in the left-hand panel
of Figure~\ref{fig_bw1}, 
the $\dot \F$ value obtained for
forcings of different frequencies falls off quickly 
for detunings $\delta \omega \gtrsim \gamma$, for
an oscillator in a coherent state.
This fall-off can be reduced by using states
with larger $\hat x$ fluctuations, such as squeezed coherent states.
The figure shows that for squeezing parameter
$G_s$, corresponding to fluctuations with variance $G_s$
times larger than that in a coherent state,
the bandwidth over which 
the $\dot \F$ bound is attained (to $\OO(1)$) is increased
by a factor $\sim G_s$.

This behaviour, shown in the left-hand panel of 
Figure~\ref{fig_bw1}, corresponds
to the QFI growth analysis in this subsection,
and puts an upper bound on how well any given measurement
scheme can do. By analysing specific measurement
schemes, we can see that this scaling can
be attained (again, to $\OO(1)$).
In the right-hand panel of Figure~\ref{fig_bw1},
we show the sensitivity of
a prepare-measure-reset scheme
as discussed in Section~\ref{sec_prepare}
(using Fock states of the oscillator),
illustrating how states with larger
$H'$ fluctuations can result
in larger sensitivity bandwidths.
In~\cite{konrad}, the sensitivity of linear amplification
schemes with squeezed coherent states is analysed, 
with analogous results.

In summary, the analysis in this section
shows that the $\sim t_c^{-1}$ scaling for the 
sensitivity bandwidth is the best achievable
(even when we do not need to determine the nuisance
parameters),
while the analyses in Section~\ref{sec_prepare} show that this scaling can
be attained. Our analysis here was rather schematic (and it would be interesting to treat
this more precisely) but captures the parametrics involved,
showing how the range of different time dependences that 
we can be sensitive to is
set by the finite-time QFI growth rate. 

\subsection{Frequent measurements}
\label{sec_prepare}

One way to increase our sensitivity bandwidth
(potentially at the expense of peak sensitivity)
is to use a prepare-measure-reset procedure:
we sequentially prepare our probe system in some known
starting state, allow it to evolve for some 
small time $t_1$, and then measure (before resetting).
If our signal changes slowly compared to $t_1$,
then it will be approximately constant over each
prepare-measure-reset
cycle, and we will have similar sensitivities to
signals with different rates of change.
Consequently, the sensitivity bandwidth
will be at least $\sim 1/t_1$.

Since the QFI bounds the classical Fisher information
from any measurement, the total information
we can obtain is bounded by the sum of the QFIs
from each cycle. Hence,
the appropriate figure of merit is
$\F(t_1)/t_1$, where $\F(t_1)$ is
the QFI at each measurement time.
This can be bounded using the techniques
from Section~\ref{secqfi}.
In particular, if we can treat
$H'$ as constant over each cycle,
then to obtain $\F(t_1)/t_1 \gtrsim \dot \F_c$,
for some $\dot \F_c$, we need
$t_1 \gtrsim \dot \F_c/\F_{H',\rm max}$,
where $\F_{H',\rm max}$ is a bound
on $\F_{H',\rm max}$.
So, if $1/t_1$ sets the sensitivity bandwidth,
then $\Delta \omega \lesssim \F_{H',\rm max}/\dot \F_c$.

A practical example of this kind of prepare-measure-reset
scheme is the proposal to use
Fock states of cavity modes for axion dark matter
detection~\cite{10.1007/978-3-030-31593-1_5,2008.12231}.
If we prepare an oscillator
in a Fock state $|n\rangle$,
then in the presence of a small (linear) forcing,
the rates of $|n\rangle \rightarrow |n+1\rangle$
and $\ket{n}\rightarrow \ket{n-1}$
transitions are $\propto n+1$ and $n$, respectively. Consequently,
$\F$ initially increases faster with larger $n$.
However, the rate of damping-induced
$|n\rangle \rightarrow |n-1\rangle$ transitions
is also $\propto n$,
so the growth of $\F$ slows down over a
timescale $\sim \gamma / n$.
This corresponds to the fact that,
for given $\gamma$,
we cannot violate the $\dot \F \le 4 |\epsilon|^2 / \gamma$
bound. However,
the faster initial growth rate
increases the sensitivity bandwidth;
quantitatively, $\langle H'^2 \rangle \propto 2n+1$,
and $\F \sim (8n + 4)(\gamma t)^2$ for small $t$.
These points are illustrated by the dashed blue
curve in the bottom panels of
Figure~\ref{fig_fq1}, which corresponds
to preparing the oscillator in
an $n=2$ Fock state.

Experimentally, one usually imagines measuring in
the Fock basis. This would not be the optimal basis
in which to measure if we knew the phase of our signal;
however, it is a simple and phase-independent choice
(and for many signals, such as virialised axion dark matter, the
phase will be unknown). For Fock basis measurements,
the time-averaged rate of Fisher information gain
is smaller than the time-averaged QFI growth rate,
 but only by an order-1 factor.
Similarly, prepare-measure-reset schemes
have smaller time-averaged QFI growth
rate than is possible with
`steady-state' schemes that
do not involve resets (assuming equivalent
mean occupation number), but this is again
an order-1 loss.
Overall, as illustrated in the right-hand
panel of Figure~\ref{fig_bw1}, one can
still use non-classical states to obtain
an enhanced sensitivity bandwidth 
using a prepare-measure-reset scheme despite these issues.

While we have mostly discussed the bandwidth over
which it is possible to achieve
near-peak sensitivity, one is often interested
in other figures of merit, such
as some kind of frequency-averaged sensitivity.
For example, if we are searching for a signal
of definite but unknown frequency, then
the appropriate quantity might be the time-averaged
$\dot \F$, averaged over the relevant frequency range.
For the prepare-measure-reset procedure above,
this is bounded
by $\sim \dot \F \times {\F_{H'}}/{\dot \F}
\sim \F_{H'}$. We obtain
parametrically the same result for steady-state schemes.
These relationships are analogous to
`Energetic Quantum Limits' derived
in the gravitational wave literature~\cite{10.1063/1.1291855,1903.09378}, which relate the frequency-averaged
sensitivity to the quantum fluctuations
of an optical mode's energy (since the energy is
the appropriate interaction operator for
a gravitational wave coupling).

\section{Conclusions}

In this paper, we explored the
problem of Hamiltonian parameter
estimation for quantum systems
subject to Markovian noise.
In particular, we derived an upper bound
(Eq.~\eqref{eqlim1}) on
the rate of increase of the quantum Fisher information,
in terms of the
Lindblad operators and the parameter derivative $H'$
of the Hamiltonian.
This bound tightens previous bounds obtained
in the case of time-independent 
$H'$ for a finite-dimensional system~\cite{prx,preskill},
and also applies directly to
more general situations, such as time-dependent
master equations, and/or infinite-dimensional systems.

For time-dependent signals, we showed
that the range of different frequencies
to which a system can have close-to-peak sensitivity
is set by the finite-time QFI growth rate,
which is bounded by the quantum fluctuations
of $H'$. 
This is true even in the large-time,
large-QFI regime, illustrating how our bounds
can be useful beyond simply studying 
early-time QFI growth.
While many previous studies have focused
on how non-classical states
can lead to higher sensitivities to specific
signals, our analysis illustrates
another way in which they can provide
metrological advantage---by expanding the sensitivity
bandwidth. This applies even in cases where
such states cannot improve
a system's peak sensitivity,
such as in the example we analysed of a damped
harmonic oscillator with linear forcing.

While such behaviour has been noted
for particular schemes, such as
force detection with squeezed coherent
states~\cite{konrad}, we demonstrated
that it applies more generally.
It may be interesting to investigate
sensitivity bandwidth expansion in gravitational
wave detection schemes using our methods;
while there have been general metrological analyses
of interferometric gravitational wave detection
in the presence of noise (in particular, photon loss)~\cite{10.1103/PhysRevA.88.041802},
these have generally focused on
the sensitivity to a specific signal.

As noted throughout Section~\ref{secqfi},
our bounds can be attained asymptotically (in the large-time limit)
using the error-correction methods
of~\cite{preskill,10.1103/PhysRevResearch.2.013235}
(at least for time-independent $H'$).
However, it is not clear whether they can always
be attained at finite times by using appropriate
ancilla-assisted schemes.
While we showed that this was possible in some
cases, such as the damped oscillator considered
in Section~\ref{secqho}, we leave the
general question to future work.

Another direction in which our results could
be extended is by considering multi-parameter
estimation problems. While Section~\ref{sec_bw}
considered very simple examples of multi-parameter
estimation problems, e.g., separate 
parameters controlling separate time intervals,
more general scenarios (such as those
analysed in~\cite{10.22331/q-2020-07-02-288}) could be investigated.

In addition to systems whose time evolution is
well-described by a Lindblad master equation, our methods
could similarly be applied to quantum channels
that are equivalent
to evolution over a finite time under some Lindblad master equation,
even if that master equation does not describe the continuous-time
evolution of the system.
This may be useful in analysing the QFI for states
obtained
from discrete applications of quantum
channels, as considered in e.g.,~\cite{10.1103/PRXQuantum.2.010343}.

\tocless\acknowledgments{We thank Masha Baryakhtar, Konrad Lehnert and Sisi Zhou for helpful
conversations. KW is supported by the Stanford Graduate Fellowship. RL's research is supported in part by the National Science Foundation under Grant No.~PHYS-2014215, and the Gordon and Betty Moore Foundation Grant GBMF7946.
RL thanks the Caltech physics department for hospitality
during the completion of this work.}

\appendix

\section{Differentiability of QFI wrt time}
\label{appL}

To derive the bounds in Section~\ref{secqfi}, we assumed that
it is always possible to find a Hermitian operator $\L(t,g)$ such
that
\begin{equation}
	\rho'(t,g) = \frac{1}{2}(\rho(t,g) \L(t,g) +
	\L(t,g)\rho(t,g))
	\label{eqrhoprime}
\end{equation} 
In this appendix, we show
that for any $\rho$ obeying a Lindblad master equation in which the operators are differentiable with respect to $g$, we can always find such an $\L$.\footnote{More generally,  this holds for any $\rho$ whose evolution is described by a quantum channel that is differentiable with respect to $g$; see Proposition~\ref{claim:rho'jk}.}  
We also prove that if $\rho$ is an analytic function
of $t$, then $\L$ is differentiable
with respect to $t$, except possibly at a
set of isolated times (for given $g$). We then show that $\F$ must be continuous
even at these isolated times. 

For a given $t$ and $g$, let $\rho = \sum_j p_j \ket{j}\bra{j}$ be a spectral decomposition of $\rho$.\footnote{We assume
that we can index the spectrum of $\rho$,
leaving the analysis for continuous
spectra to future work.} 
Then, Eq.~\eqref{eqrhoprime} implies that for all $j,k$,
\begin{equation} 
\bra{j}\rho'\ket{k} = \frac{1}{2}(p_j + p_k)\bra{j}\L\ket{k}. 
\end{equation}
Hence, we set 
\begin{equation}
    \bra{j}\L\ket{k} = \frac{2\bra{j}\rho'\ket{k}}{p_j + p_k}
	\label{eqLjk}
\end{equation}
for all $(j,k)$ such that $p_j + p_k \neq 0$, while for $(j,k)$ such that $p_j = p_k = 0$, we choose any arbitrary values for $\bra{j}\L\ket{k}$ (that are consistent with Hermiticity). Clearly, this $\L$ will satisfy Eq.~\eqref{eqrhoprime} provided that $\bra{j}\rho'\ket{k} = 0$ for all $(j,k)$ such that $p_j = p_k = 0$. The following proposition shows that if this condition on $\rho'$ is satisfied at some initial time, then it is satisfied at all subsequent times.

\begin{prop} \label{claim:rho'jk}
Suppose that for some $t_0, t_1$ such that $t_0 < t_1$, $\rho(t_1,g)$ is obtained from $\rho(t_0,g)$ via a quantum channel:
\begin{equation} \label{eq:channel}
\rho(t_1,g) = \sum_m M_m(g)\rho(t_0,g) M_m^\dagger(g),
\end{equation}
where the Kraus operators $M_m$ are differentiable with respect to $g$. For some fixed value of $g$, let $\rho(t_0,g) = \sum_j p_j \ket{j}\bra{j}$ and $\rho(t_1,g) = \sum_J P_J \ket{J}\bra{J}$ be spectral decompositions of $\rho(t_0,g)$ and $\rho(t_1,g)$. If $\bra{j}\rho'(t_0,g)\ket{k} = 0$ for all $(j,k)$ such that $p_j = p_k = 0$, then $\bra{J}\rho'(t_1,g)\ket{K} = 0$ for all $(J,K)$ such that $P_J = P_K = 0$. 
\end{prop}

\begin{proof} For clarity, we will denote $\rho_0 \equiv \rho(t_0,g)$ and $\rho_1 \equiv \rho(t_1,g)$, and
omit the $t$ and $g$ arguments of all operators. First, observe that for any $J$ such that $P_J = 0$, 
\begin{align}
    0 &=\bra{J}\rho_1 \ket{J} \nonumber \\
    &= \sum_m \bra{J}  M_m \rho_0 M_m^\dagger \ket{J} \nonumber \\
    &= \sum_{j: p_j \neq 0} p_j \sum_m \left|\bra{J} M_m \ket{j} \right|^2. \label{rhojj}
\end{align}
This implies that 
\begin{equation} \label{eq:MJj}
\bra{J}M_m\ket{j} = 0
\end{equation}
for all $m$ for any $(j,J)$ such that $p_j \neq 0$ and $P_J = 0$. 

Differentiating Eq.~\eqref{eq:channel} with respect to $g$, we obtain
\begin{align}
    &\bra{J}\rho_1'\ket{K} \nonumber \\ &= \bra{J}\sum_m\Big(M_m' \rho_0 M_m^\dagger + M_m \rho_0 M_m^{\prime\dagger}  + M_m \rho_0' M_m^\dagger \Big)\ket{K} \nonumber \\
    &= \sum_m \Bigg[\sum_{j:p_j \neq 0} p_j \Big(\bra{J}M_m'\ket{j}\bra{j}M_m^\dagger \ket{K} \nonumber \\
    &\qquad \qquad\qquad \enspace + \bra{J}M_m\ket{j}\bra{j}M_m^{\prime\dagger}\ket{K}\Big) \nonumber\\
    &\qquad\quad + \sum_{j,k} \bra{J}M_m \ket{j}\bra{j}\rho_0'\ket{k}\bra{k}M_m^\dagger\ket{K} \Bigg]. \label{eq:prop1sum}
\end{align}
Consider any $(J,K)$ such that $P_J = P_K = 0$. Then, the first two terms in Eq.~\eqref{eq:prop1sum} 
vanish by the result expressed in Eq.~\eqref{eq:MJj}. For the third term, if $(j,k)$ is such that $p_j = p_k = 0$, then $\bra{j}\rho_0'\ket{k} = 0$ 
by assumption; otherwise, at least one of $\bra{J}M_m\ket{j}$ or $\bra{k}M_m^\dagger\ket{K}$ 
is zero by Eq.~\eqref{eq:MJj}. Thus, the third term also vanishes. Therefore, $\bra{J}\rho_1'\ket{K} = 0$ as claimed.

\end{proof}

Evolution via a Lindblad master equation
allows us to write $\rho(t_1,g)$ as a quantum
channel on $\rho(t_0,g)$ for any $t_1 > t_0$,
fulfilling the assumption in Eq.~\eqref{eq:channel} of Proposition~\ref{claim:rho'jk}.
Consequently, if we can find $\L(t_0,g)$ at some
initial time $t_0$ (e.g., if
we prepare our system in a $g$-independent state, so $\rho'(t_0,g) = 0$ at the start time $t_0$), then
we can always find $\L(t,g)$ satisfying Eq.~\eqref{eqrhoprime}
at all subsequent times $t$, using Eq.~\eqref{eqLjk}.

For our analyses in Section~\ref{secqfi}, we
also assumed that $\L$ is differentiable with
respect to $t$, except at a
(possibly empty) set of isolated times. Writing
$\L = \sum_{j,k} \L_{jk}|j\rangle \langle k |$,
a sufficient condition is that $\L_{jk}, |j\rangle$
and $|k\rangle$ are differentiable with respect to $t$
for all $j,k$. 

If we consider arbitrary quantum channels,
then it is simple to write down a channel
such that $\L$ is \emph{not} differentiable with respect to $t$,
so we need to impose more conditions.
Here, we analyse, as an example, the simple case where
$\rho(t,g)$ is an analytic function of $t$ (for
each $g$); this often serves as a good model for
physical systems (e.g.,
if $\rho$ arises from a master equation whose
operators are analytic functions of $t$,
as in Section~\ref{secqho}).
In this case, the eigenvalues $p_j(t)$
and eigenstates $\ket{j(t)}$ can be
chosen to be
analytic functions of $t$~\cite{kato}.\footnote{This result applies to finite-dimensional systems; 
we
defer a careful analysis of the differentiability assumption for infinite-dimensional
systems to future work.}
Hence, it remains to show that
$\L_{jk}$ is differentiable with respect to $t$.

For given $(j,k)$, at any $t$ such that $p_j + p_k \neq 0$, Eq.~\eqref{eqLjk}
is clearly differentiable, since
the $p_j$ and $\ket{j}$ are differentiable.
Moreover, by analyticity, the set
of points with $p_j + p_k = 0$ (for some given $g$) is either the entire $t$ range, or a set
of isolated points. In the latter case, we can split
our $\L$ evolution into differentiable segments
between these points.
If $\F$ is continuous at these
points, then $\F(t)$ for all $t$ can be obtained by
integrating $\dot \F$ in the segments
between the points.
To show that $\F$ is continuous, we prove
a strengthened version of Proposition~\ref{claim:rho'jk}.

\begin{prop} \label{claim:rho'jk2}
Suppose that for some $t_0, t_1$ such that $t_0 < t_1$, $\rho(t,g)$ is obtained from $\rho(t_0,g)$ via a quantum channel, for all $t$ in a neighbourhood $I_1$ of $t_1$:
\begin{equation} \label{channel2}
\rho(t,g) = \sum_m M_m(t,g) \rho(t_0,g) M_m^\dagger(t,g),
\end{equation}
where the Kraus operators $M_m(t,g)$ are differentiable with respect to $g$ and $M_m'(t,g)$ are continuous in $t$. For some fixed value of $g$, let $\rho(t,g) = \sum_j p_j(t) \ket{j(t)}\bra{j(t)}$ 
denote the spectral decomposition of $\rho(t,g)$ for $t \in I_1$. Assume that $\bra{j(t_0)}\rho'(t_0,g)\ket{k(t_0)} = 0$ for all $(j,k)$ such that $p_j(t_0) = p_k(t_0) = 0$. If for some $(j,k)$, we have $p_j(t) + p_k(t) = \mathcal{O}(f(t - t_1)^2)$ (for $|t - t_1|$ small) 
for some function $f = o(1)$,\footnote{since
we are interested in $t_1$ such that
$p_j(t_1) + p_k(t_1) = 0$} then $\bra{j(t)}\rho'(t,g)\ket{k(t)} = \mathcal{O}(f(t-t_1))$.
\end{prop}
\begin{proof} The proof is similar to that of Proposition~\ref{claim:rho'jk}. For convenience, we will sometimes omit the argument $g$, 
and we denote $\rho_0 \equiv \rho(t_0,g)$ and $(A)_{Jj}(t) \equiv \bra{J(t)}A(t)\ket{j(t_0)}$ 
for any operator $A$ and indices $J,j$ (ranging over the eigenbasis of $\rho$).

Using the same argument that led to Eq.~\eqref{rhojj}, we have that for any $J$,
\begin{equation}
    p_J(t) = \sum_{j:p_j(t_0) \neq 0} p_j(t_0) \sum_m \left|(M_m)_{Jj}(t)\right|^2
\end{equation}
for all $t \in I_1$. Since $p_j(t_0)$ are non-negative constants independent of $t$, we see that if $p_J((t) = \mathcal{O}(f(t-t_1)^2)$, then 
\begin{equation} \label{MmJjt}
(M_m)_{Jj}(t) = \mathcal{O}(f(t-t_1))
\end{equation} for all $j$ such that $p_j(t_0) \neq 0$. 

	Also, it follows from $\sum_m M_m^\dagger M_m = I$~\cite{nc} that $\sum_m \sum_j |(M_m)_{Jj}(t)|^2 = 1$, so 
\begin{equation} \label{MmJjt1} |(M_m)_{Jj}(t)| \leq 1
\end{equation} for all $m, J,j$.

Differentiating Eq.~\eqref{channel2} with respect to $g$, we have for any $J,K$,
\begin{align}
    &\bra{J(t)}\rho'(t,g)\ket{K(t)} \nonumber \\
    &= \sum_m \Bigg[\sum_{j:p_j(t_0) \neq 0} p_j \Big((M_m')_{Jj}(t)(M_m^\dagger)_{jK}(t) \nonumber \\
    &\qquad \qquad\qquad \enspace + (M_m\ket{j})_{Jj}(M_m^{\prime\dagger})_{jK}(t)\Big) \nonumber\\
    &\qquad\quad + \sum_{j,k} (M_m)_{Jj}(t)\bra{j(t_0)}\rho_0'\ket{k(t_0)}(M_m^\dagger)_{kK}(t) \Bigg]. \label{eq:prop2sum}   
\end{align}
Since $M_m'$ is continuous in $t$, $(M_m')_{Jj}(t), (M_m'^\dagger)_{jK}(t) = O(1)$ for $|t - t_1|$ small. Hence, for any $(J,K)$ such that $p_J(t) + p_K(t) = \mathcal{O}(f(t-t_1)^2)$, the first two terms in Eq.~\eqref{eq:prop2sum} are both $\mathcal{O}(f(t-t_1))$, by the result expressed in Eq.~\eqref{MmJjt}. As for the third term, if $(j,k)$ is such that $p_j(t_0) = p_k(t_0) = 0$, then $\bra{j(t)}\rho_0'\ket{k(t)} = 0$ by assumption; otherwise, at least one of $(M_m)_{Jj}(t)$ and $(M_m^\dagger)_{kK}(t)$ is $\mathcal{O}(f(t-t_1))$ by Eq.~\eqref{MmJjt}, and the other is $\mathcal{O}(1)$ by Eq.~\eqref{MmJjt1} (and $\bra{j(t_0)}\rho'_0 \ket{k(t_0)}$ is independent of $t$). 

\end{proof}

The QFI is given by
\begin{equation}
	\F = \tr(\rho \L^2) = 2 \sum_{\substack{j,k:\\p_j + p_k \neq 0}}\frac{|\bra{j}\rho'\ket{k}|^2}{p_j + p_k}
	\label{eq_fqi}
\end{equation}
If the denominator $p_j(t) + p_k(t) = \OO(f(t - t_1)^2)$ for some $t_1$ and $(j,k)$, then
Proposition~\ref{claim:rho'jk2} shows that the numerator $|\bra{j(t)}\rho'(t,g)\ket{k(t)}|^2 = \OO(f(t - t_1)^2)$. 
Thus, since $p_j,p_k$ and $\bra j 
\rho' \ket k$ are all analytic in $t$, the Taylor series 
of the numerator and denominator around
$t = t_1$ have leading terms of the same order
in $t - t_1$.
As a result, any apparent singularities in
the RHS terms of Eq.~\eqref{eq_fqi} are removable,
so $\F$ is continuous in $t$.

The continuity of the QFI, with respect to
different variables, has been investigated 
in a number of papers~\cite{10.1103/PhysRevA.94.012339,10.1103/PhysRevA.95.052320,10.1103/PhysRevA.100.032317}; 
as far as we are aware, Proposition~\ref{claim:rho'jk2} constitutes 
a new contribution on this topic, which may be of independent interest.

\section{Properties of $\F_A$}
\label{appqfi}

In this appendix, we review some properties
of the quantum Fisher information $\F_A$ with respect to a Hermitian operator $A$, 
defined by Eqs.~\eqref{eqLG} and~\eqref{eq:FA}. (Note from Eq.~\eqref{eq:FA} that $\F_A$ also depends on the state $\rho$.)

\begin{fact} \label{fact:LA}
For any Hermitian operator $A$ and density operator $\rho$, there exists a Hermitian operator $\L_A$ satisfying Eq.~\eqref{eqLG}.
\end{fact}
\begin{proof}
Let $\rho = \sum_j p_j \ket{j}\bra{j}$ be a spectral decomposition of $\rho$. Then, set
\begin{equation} \label{LAjk} \bra{j}\L_A\ket{k} = \frac{2\bra{j}i[\rho,A]\ket{k}}{p_j + p_k}\end{equation}
for all $(j,k)$ for which $p_j + p_k \neq 0$, while for $(j,k)$ such that $p_j = p_k = 0$, choose any value for $\bra{j}\L \ket{k}$ (consistent with Hermiticity). This satisfies Eq.~\eqref{eqLG} since $\bra{j}i[\rho,A]\ket{k} = 0$ for all $(j,k)$ such that $p_j = p_k = 0$.

\end{proof}

\begin{fact} \label{fact:FaA+bI}
For any Hermitian operator $A$, density operator $\rho$, and $a,b\in \mathbb{R}$, we have
\begin{equation} \F_{a A + b I} = a^2 \F_A. \end{equation}
\end{fact}
\begin{proof}
From Eq.~\eqref{LAjk} in the proof of Fact~\ref{fact:LA}, we can choose $\L_{aA + bI} = a\L_A$, so $\F_{aA + bI} = \tr(\rho\L_{aA + bI}^2) = \tr(\rho(a\L_A)^2) = a^2\F_A$ by Eq.~\eqref{eq:FA}.

\end{proof}

\begin{fact}[\cite{10.1088/1751-8113/47/42/424006}, Equations~60 and~61] \label{fact:FVarA}
For any Hermitian operator $A$ and density operator $\rho$, we have
\[ \F_A(\rho) \leq \tr(\rho A^2) - \tr(\rho A)^2 = \Var_\rho(A), \]
with equality if $\rho$ is pure.
\end{fact}

\begin{fact} \label{fact:FA+B}
For any Hermitian operators $A,B$ and density operator $\rho$, we have
\begin{equation}
    \F_{A+B} \leq \left(\sqrt{\F_A} + \sqrt{\F_B} \right)^2,
\end{equation}
with equality iff $\sqrt{\rho}\L_A = c \sqrt{\rho}\L_B$ for $c \geq 0$ or $\sqrt{\rho}\L_B = 0$.
\end{fact}
\begin{proof}
Note that $\F_{A} = \|\sqrt{\rho}\L_A\|^2_2$, with $\|\cdot\|_2$ the Hilbert-Schmidt norm. From Eq.~\eqref{LAjk}, $\L_{A + B} = \L_A + \L_B$. Hence,
\begin{align}
    \F_{A + B} &= \|\sqrt{\rho}\mathcal{L}_{A+B}\|_2^2 \nonumber \\
    &= \|\sqrt{\rho}(\L_A + \L_B)\|_2^2 \nonumber \\
    &\leq \Big(\|\sqrt{\rho}\L_A\|_2 + \|\sqrt{\rho}\L_B\|_2 \Big)^2 \nonumber\\
    &= \left(\sqrt{\F_A} + \sqrt{\F_B}\right)^2,
\end{align}
and the triangle inequality used in the third line is saturated iff $\sqrt{\rho}\L_A = c \sqrt{\rho}\L_B$ for $c \geq 0$ or $\sqrt{\rho}\L_B = 0$.

\end{proof}

\section{Details for Section~\ref{sec_compare}}
\label{appComp}

In this appendix, we fill in some of the details for Section~\ref{sec_compare}. We start by reviewing the QFI calculations in~\cite{prx,preskill}, which lead to the bound in Eq.~\eqref{eqlinbound} for the time-independent, HLS case (Appendix~\ref{appC1}). We then show (Appendix~\ref{appC2}) that the best bound that can be obtained for the HNLS case using the methods in~\cite{prx,preskill} is Eq.~\eqref{eqfqtpreskill}. Finally, we provide the technical details leading to our HNLS bound in Eq.~\eqref{eqfhnls} (Appendix~\ref{appC3}).

\toclesslab\subsection{Review of previous results}{appC1}

The bounds on the QFI in
\cite{prx,preskill}, which apply in the case
of time-independent $H'$ and $L_j$, were derived from formulae
for the QFI for $N$ identical operations~\cite{10.1103/PhysRevLett.113.250801,10.1088/1751-8113/41/25/255304}, taking the limit in which $N \to \infty$, so that the time interval $dt = t/N \to 0$ for each operation becomes
correspondingly short.
The form given in
\cite{preskill} (\cite{prx} gives a similar expression) is
\begin{equation}
	\F(t)\le 4 \frac{t}{dt} \|\alpha_{dt}\| +
	4 \left(\frac{t}{dt}\right)^2 \|\beta_{dt}\|
	\left(\|\beta_{dt}\| + 2 \sqrt{\|\alpha_{dt}\|}\right),
	\label{eq_preskillf}
\end{equation}
where 
\begin{equation}
    \alpha_{dt} \coloneqq \sum_j\Big(K'_j -i \sum_k h_{jk} K_k\Big)^\dagger \Big(K'_j - i\sum_l h_{jl}K_l\Big)
\end{equation}
and
\begin{equation}
\beta_{dt} \coloneqq i\sum_j \Big(K'_j - i\sum_k h_{jk} K_k\Big)^\dagger K_j,
\end{equation}
with $K_0 = I -(igH' +\frac{1}{2}\sum_{j\geq 1} L_j^\dagger L_j)dt$ and $K_j = L_j\sqrt{dt}$ for $j \geq 1$ the Kraus operators describing the time evolution over each interval of length $dt$, and $h$ an arbitrary Hermitian matrix. The RHS of Eq.~\eqref{eq_preskillf} can be minimised over the choice of $h$. 

Write $X = X^{(0)} + X^{(1)}\sqrt{dt} + X^{(2)}dt  + \dots$ as the expansion of any quantity $X$ in powers of $\sqrt{dt}$.
To obtain a sensible bound from Eq.~\eqref{eq_preskillf}
in the $dt \rightarrow 0$ limit, we need
that 
\begin{align} &\alpha_{dt}^{(0)} = \alpha_{dt}^{(1)} = 0, \nonumber \\
&\beta_{dt}^{(0)} = \beta_{dt}^{(1)} = 0.  \label{alphabeta0}
\end{align}
Moreover, if $\beta_{dt}^{(2)} \neq 0$, then we
need that $\alpha_{dt}^{(2)} = 0$
for the $8 t^2 \|\beta_{dt}^{(2)}\| \sqrt{\|\alpha_{dt}\|} / dt$
term to not blow up.

As shown in~\cite{preskill}, $\alpha_{dt}^{(0)} = 0$ $\Leftrightarrow$ $h_{0j}^{(0)} = 0 \enspace \forall\, j$ $\Rightarrow$ $\alpha_{dt}^{(1)} = \beta_{dt}^{(0)} = 0$, in which case $\beta_{dt}^{(1)} = 0$ $\Leftrightarrow$ $-h_{00}^{(1)} = 0$. Then, under these conditions, we have\footnote{This differs from Equation~(52) in~\cite{preskill} slightly due to some small typos therein.}
\begin{align} 
\beta_{dt}^{(2)} &= -H' - h_{00}^{(2)} I \nonumber \\
&\quad - \sum_{j \geq 1}\left(h_{0j}^{(1)}L_j + h_{0j}^{(1)*}L_j^\dagger\right)
- \sum_{j,k \geq 1} h_{jk}^{(0)}L_j^\dagger L_k.
\end{align}
We see that this has the form of Eq.~\eqref{eq_hdecomp} if one makes the notational substitutions
\begin{align}
    &G \leftrightarrow -\beta^{(2)}_{dt}, \quad \alpha \leftrightarrow - h_{00}^{(2)}, \\ \quad &\beta_j \leftrightarrow -h_{j0}^{(1)}, \quad \gamma_{jk} \leftrightarrow - h_{jk}^{(0)}. \label{notation}
\end{align}
It can also be checked that $\alpha_{dt}^{(2)} = \sum_j A_j^\dagger A_j$ when Eq.~\eqref{alphabeta0} is satsified, where $A_j$ is defined as in Eq.~\eqref{eq:A_j}. 

\cite{preskill} uses Eq.~\eqref{eq_preskillf} to derive a bound for the HLS case, where $\beta^{(2)}_{dt}$ can be set to $0$. They show that with this choice, Eq.~\eqref{eq_preskillf} reduces in the $dt \to 0$ limit to Eq.~\eqref{eqlinbound}, which is the bound we compare to in section~\ref{sec_hls}.

\toclesslab\subsection{HNLS bound derived from Eq.~\eqref{eq_preskillf}}{appC2}

We now show that the best possible bound given by Eq.~\eqref{eq_preskillf} for the HNLS case has the form in Eq.~\eqref{eqfqtpreskill}. Since $H'$ is not in the Lindblad span, $\beta^{(2)}_{dt}$ cannot be set to zero, so as noted above, we need to have $\alpha_{dt}^{(2)} = 0$. This holds iff $A_j = 0$ for all $j$, i.e., (in our notation)
\begin{equation} \label{Aj01}
    \sum_k \gamma_{jk}L_k = -\beta_j I
\end{equation}
by Eq.~\eqref{eq:A_j}, which implies that for all $j$,
\begin{equation} \label{Aj02}
    \sum_k \gamma_{jk}L_k = \sum_k \gamma_{jk} \bra{\psi}L_k\ket{\psi}
\end{equation}
for any $\ket{\psi}$. We then have, for arbitrary $\ket{\psi}$,
\begin{align}
    &\sum_j (\beta_j^* L_j + \beta_j L_j^\dagger) + \sum_{j,k}\gamma_{jk}L_j^\dagger L_k \nonumber \\
    &\quad = \sum_{j,k}\gamma_{jk}L_j^\dagger L_k \nonumber  \\
    &\quad = \sum_j L_j^\dagger \sum_k \gamma_{jk}\bra{\psi}L_k\ket{\psi} \nonumber \\
    &\quad = \sum_k (-\beta_k^* I) \bra{\psi}L_k \ket{\psi},
\end{align}
using Eqs.~\eqref{Aj01} and~\eqref{Aj02} along with $\gamma_{jk} = \gamma_{kj}^*$. This shows that $\sum_j(\beta_j^* L_j + \beta_j L_j^\dagger) + \sum_{j,k}\gamma_{jk}L_j^\dagger L_k$ is proportional to the identity. Therefore, from Eq.~\eqref{eq_hdecomp}, $G$ must be of the form
\begin{equation}
    G = H' - \overline{\alpha}I
\end{equation}
for some $\overline{\alpha} \in \mathbb{R}$. Eq.~\eqref{eqfqtpreskill} then follows from Eq.~\eqref{eq_preskillf} by noting that $G$ corresponds to $-\beta_{dt}^{(2)}$ (under the
conditions
in Eq.~\eqref{alphabeta0})
and taking the minimum over $\overline{\alpha}$.

\bigskip

\bigskip

\toclesslab\subsection{HNLS bound derived from Eq.~\eqref{eqlim0}}{appC3}

In Section~\ref{sec_hnls}, we derived bounds
on $\F(t)$ in the generic case where $H'$ and $L_j$ are time-dependent, and $H'$ may or may not be in the Lindblad span at different times.
Assuming time-independent bounds $c_0, c_1, c_2$
on the relevant operators, with $c_1 \geq c_0$ (cf.~Section~\ref{sec_hnls}),
the bound for $t \geq t_c \coloneqq \frac{2c_2}{(c_1 - c_0)^2}\ln\left(\frac{2 c_1}{c_1 + c_0}\right)$ is given in Eq.~\eqref{eqfhnls} as $\F(t) \leq y(t-t_c)^2$, where the function $y$ is defined as
\begin{widetext}
\begin{equation} \label{y}
    y(t) \coloneqq \frac{c_2}{c_0}\left[-W_{-1}\Bigg(-\frac{1}{e}\left(\frac{c_0+c_1}{c_1 - c_0}\right)\exp\left( -\frac{2c_0^2}{c_2}t- \frac{2c_0}{c_1-c_0}\right) \Bigg)-1 \right]
\end{equation}
\end{widetext}
with
$W_{-1}$ the lower branch of the Lambert-W
function~\cite{10.1007/BF02124750}.
From~\cite{10.1109/LCOMM.2013.070113.130972}, for any $u > 0$,
\begin{equation} \label{W-1bound}
    1 + \sqrt{2u} + \frac{2}{3}u \leq -W_{-1}(-e^{-1-u}) \leq 1 + \sqrt{2u} + u.
\end{equation}
We can write $y(t) = \frac{c_2}{c_0}[- W_{-1}(-e^{-1-u(t)})-1]$
with
\begin{equation}
    u(t) \coloneqq \frac{2c_0^2}{c_2} t + \frac{2c_0}{c_1 - c_0} - \ln\left(\frac{c_0+c_1}{c_1-c_0}\right),
\end{equation}
which is non-negative for all $t \geq 0$, since $c_1 \geq c_0$ and $c_2 \geq 0$. Thus, by Eq.~\eqref{W-1bound}
\begin{align}
    \frac{c_2^2}{c_0^2}\left[\sqrt{2u(t)} +\frac{2}{3}u(t) \right]^2\leq y(t)^2 &\leq \frac{c_2^2}{c_0^2}\left[\sqrt{2u(t)} + u(t) \right]^2,
\end{align}
so since $u(t) = 2\frac{c_0^2}{c_2}t + \OO(1)$, 
	\begin{equation}
	y(t)^2 = 4c_0^2 t^2\left(1 + \OO(t^{-1/2})\right)
	\end{equation}
for large $t$. We have $\F(t) \leq y(t-t_c)^2$, so this
leads to Eq.~\eqref{eqFtscaling}, confirming the expected scaling.

We can also derive a simpler but weaker
bound than Eq.~\eqref{eqfhnls}, by
choosing $\delta(t) = 1$ for all times $t$
in Eq.~\eqref{eq34}, rather
than choosing the optimal $\delta$ at
each $t$. In that case, we obtain $\dot \F \le 4(c_0\sqrt\F + c_2)$,
giving (for $\F(t=0)= 0$)
\begin{align}
    \F(t) &\leq \frac{c_2^2}{c_0^2}\left[-W_{-1}\Bigg(-\frac{1}{e}\exp\left(-\frac{2c_0^2}{c_2}t\right) \Bigg) - 1\right]^2 \nonumber \\
    &\leq 4c_0^2 t^2\left(1+\frac{1}{c_0}\sqrt{\frac{c_2}{t}}\right)^2 \nonumber \\
    &= 4c_2 t + 4t^2 c_0 \left(c_0 + 2\sqrt{\frac{c_2}{t}}\right). \label{eqfhnls2}
\end{align}
Thus, we obtain a bound on $\F(t)$
that looks rather similar to Eq.~\eqref{eq_preskillf};
an important difference is that it still behaves sensibly
if $c_0$ and $c_2$ (which are analogous to $\|\beta_{dt}^{(2)}\|$ and $\|\alpha_{dt}^{(2)}\|$ in the context of Eq.~\eqref{eq_preskillf}) are both non-zero.
Of course, Eq.~\eqref{eqfhnls2} is looser
than Eq.~\eqref{eqfhnls}; while this looser bound has the correct large-$t$ scaling, it
is not tight at small $t$ (unlike Eq.~\eqref{eqfhnls}).

\section{Lindblad parameter estimation}
\label{sec_lindblad_app}

In the main text, we assumed that only
the Hamiltonian depends on our parameter $g$.
This is a good model for many signal detection problems,
in which a small, effectively classical influence acts
on the detection system. However, we may also be interested
in determining properties of the system-environment 
coupling, such as the temperature of a thermal
environment, or the strength of the coupling.
This can be modelled by estimating
a parameter controlling the Lindblad operators.

Hence, in this appendix, we consider the most general case where both the Hamiltonian
$H$ and the Lindblad operators $L_j$ depend on $g$.
By substituting Eq.~\eqref{eqsme} into the expression for $\dot{\F}$ given in Eq.~\eqref{eq_dfq2} and simplifying using Eq.~\eqref{eq_ldefn}, we obtain
\begin{align}
    \dot{\F} &= 2i\tr\left(\rho[H',\L]\right) - \sum_j \tr\left(\rho [L_j,\L]^\dagger [L_j,\L]\right)\nonumber \\ 
    &-2 \sum_j \mathrm{Re}\left\{\tr\left(\rho\big(L_j^{\prime\dagger}[L_j,\L] + L_j^\dagger[L_j',\L]\big) \right) \right\}. \label{FdotL}
\end{align}

If we allow ourselves complete freedom to choose 
the Lindblad operators, it is easy to see that,
even if $H' = 0$ here, we can obtain similarly complicated
behaviour to the $g$-dependent $H$ case considered in the main text.
In particular, 
suppose that we have one Lindblad operator $L_1$ and that at some $g = g_0$, we have $L_1 = I$ and $L_1^{\prime \dagger} = iG$ for some Hermitian operator $G$. Then, the third term in Eq.~\eqref{FdotL} is equal to $2 i \tr(\rho [G,\L])$, so $\dot{\F}$ would have the same form as our expression in Eq.~\eqref{eqfqdot} for the case of $g$-independent Lindblad operators, except with $H'$ replaced by $H' + G$. Consequently, for different choices
of $G$, we can obtain all of the different behaviours
studied in the main text. In particular, if $H' + G$ is
not in the Lindblad span, then $\dot \F$ can be
arbitrarily large,
for appropriate $\rho$ and $\L$ (corresponding to $\F$
being able to grow faster than linearly, as in the HNLS case).

This $L_1 = I$ example is somewhat artificial, since
a Lindblad operator that is proportional to the identity
has no effect on the master equation. In some sense,
the behaviour described above arises from a non-canonical choice
of Lindblad operators. However, we can show that
even for a canonical parameterisation of the Lindblad
terms, similar behaviour can still arise (in particular,
$\dot \F$ can still become arbitrarily large).

A canonical way
of writing the master equation for finite-dimensional systems is
\begin{equation} \label{canonical}
\dot{\rho} = -i[H,\rho] + \sum_{k,l}h_{kl}\Big(M_k\rho M_l^\dagger - \frac{1}{2}\left\{M_l^\dagger M_k,\rho\right\}\Big),
\end{equation}
where $h$ is a positive semidefinite matrix and $\{M_k\}_k$ is a fixed orthonormal basis for traceless operators.
This is canonical
in the sense that different choices of $h$ correspond
to physically different master equations, whereas
different choices of Lindblad operators $\{L_j\}_j$ in Eq.~\eqref{eqsme} can
give rise to the same master equation. To arrive at this form, one takes each Lindblad operator $L_j$ in Eq.~\eqref{eqsme} to be traceless wlog (modifying $H$ if necessary) and decomposes it in the basis $\{M_k\}_k$. 

Substituting Eq.~\eqref{canonical} into Eq.~\eqref{eq_dfq2} gives 
\begin{align}
    \dot{\F} 
    &= 2i\tr(\rho [H',\L]) - \sum_{k,l} h_{kl}\tr\left(\rho[M_l,\L]^\dagger[M_k,\L]\right)\nonumber \\
    &- 2\sum_{k,l}\mathrm{Re}\Big\{h_{kl}' \tr\left(\rho M_l^\dagger[M_k,\L]\right) \Big\}. \label{FdotM}
\end{align}
Since $h$ is positive semidefinite, we can write $h = s^2$ for some Hermitian matrix $s$, and we take $s$ to be differentiable with respect to $g$. Then,
\begin{align}
    &-2\sum_{k,l}\mathrm{Re}\Big\{h_{kl}'\tr\left(\rho M_l^\dagger [M_k,\L]\right)\Big\} \nonumber\\
    & = -2\sum_j \mathrm{Re}\Bigg\{ \tr\Bigg(\rho \Big(\sum_l s_{lj}' M_l\Big)^\dagger \Big[\sum_k s_{kj}M_k,\L\Big]\Bigg)\Bigg\} \nonumber \\
    &\quad\, -2\sum_j \mathrm{Re} \Bigg\{\tr\Bigg(\rho \Big(\sum_l s_{lj}M_l\Big)^\dagger \Big[\sum_k s_{kj}'M_k,\L\Big] \Bigg) \Bigg\}. \label{FdotM3}
\end{align}
Noting that $\{\sum_k s_{kj} M_k\}_j$ is a valid set of Lindblad operators for Eq.~\eqref{canonical}, the first term of the RHS of Eq.~\eqref{FdotM3} can be combined with the first two terms in Eq.~\eqref{FdotM}, and then the resulting expression can be upper-bounded using a similar argument as that leading to Eq.~\eqref{quadratic_bound} in the main text. 
In particular, if $H'$ is in the Lindblad span,
then the resulting bound is $\L$-independent.
However, the second term in Eq.~\eqref{FdotM3} cannot always be bounded in this
way. 
Specifically, if $s$ is singular, and $s'$ does not map
the kernel of $s$ to itself, then for some $\rho$,
we can make $\dot \F$ arbitrarily large by 
choosing $\L$ appropriately.

As a simple example, consider a two-level system, with a single Lindblad operator
$L_1(g) = \sqrt \gamma (\cos g \,\sigma_z
+ \sin g \, \sigma_y)$ (with $\gamma > 0$),
so that $g$ parameterizes the `direction' of the dephasing
on the Bloch sphere. 
At $g =0$, 
\begin{align} \label{eq_gdirection}
	\dot \F &= 2 i \tr (\rho [H',\L])
	+ 2 \gamma \Real \tr (\rho[\sigma_z,\L]
	\sigma_y + [\sigma_y,\L] \sigma_z) \nonumber \\
	&- \gamma \tr(\rho[\sigma_z,\L]^\dagger [\sigma_z,\L]),
\end{align}
so if we have $\L = \alpha I + \beta \sigma_z$ for some $\alpha,\beta\in\mathbb{R}$,
then $\dot \F = 2 \beta(i \tr(\rho [H',\sigma_z])
+ 2 \gamma \tr (\rho \sigma_y))$.
Consequently, even if $H'$ is in the Lindblad span (so that $[H',\sigma_z] = 0$),
we can have arbitrarily large $\dot \F$ if $\tr(\rho\sigma_y) \neq 0$.

However, for more restricted forms of Lindblad parameter dependence, 
there do exist $\L$-independent
bounds on $\dot \F$. In particular,
if $g$ only affects the magnitude of the Lindblad operators, in the sense that $L_j = f_j(g) \hat L_j$ for some constant
operators $\hat L_j$, with $f_j$ real wlog, then from Eq.~\eqref{FdotL}, we have
(for $H'=0$), 
\begin{align}
    \dot{\F} &= -4 \sum_j \mathrm{Re}\left\{\tr\left( f_j f_j'\tr\big(\rho \hat{L}_j^\dagger [\hat{L}_j,\L]\big) \right) \right\}\nonumber\\
    &\quad- \sum_j f_j^2 \tr\big(\rho [\hat{L}_j^\dagger,\L]^\dagger[\hat{L}_j,\L]\big) \\
    &\leq 4\sum_j \Bigg(|f_j'|\sqrt{\tr(\rho\hat{L}_j^\dagger \hat{L}_j)} \sqrt{f_j^2 \tr(\rho[\hat{L}_j,\L]^\dagger[\hat{L}_j,\L])} \nonumber\\
    &\quad - \frac{1}{4} f_j^2\tr(\rho[\hat{L}_j,\L]^\dagger [\hat{L}_j,\L])\Bigg) \\
    &\leq 4\sum_j (f_j')^2 \tr\left(\rho \hat{L}_j^\dagger \hat{L}_j\right) \label{eq_dflj}
\end{align}
by the same logic as that leading to Eq.~\eqref{quadratic_bound} (non-zero $H'$ can be handled straightforwardly
via a very similar calculation).
This illustrates an important difference from
the Hamiltonian parameter estimation case studied in the main text, where we found that even if $H' \propto H$, it is not necessarily in the Lindblad
span, so there may not be a $\L$-independent bound on
$\dot \F$. In contrast, if $L_j' \propto L_j$,
then Eq.~\eqref{eq_dflj} gives such a $\L$-independent bound.
Thus, for estimating e.g., the strength of
a specific coupling to the environment, or
the temperature of a thermal bath [cf.~Eq.~\eqref{eq_oscsme}], $\F$ can grow
at most linearly at large times (given
a time-independent bound on $\langle \hat L_j^\dagger \hat L_j \rangle$).

Moreover, unlike in Hamiltonian parameter estimation,
for which $\F \sim t^2$ at small $t$
if we start from $\F(t=0) = 0$ (Section~\ref{secshort}), for master equations with $g$-dependent Lindblads,
it is possible for $\dot \F$ to be nonzero
even initially.
In some circumstances, we can attain
the bound from Eq.~\eqref{eq_dflj} immediately.
For example, suppose that we have a single
Lindblad operator $L_1$, with $g$-dependent magnitude.
Then, if we start in a pure state $\ket \psi$, 
and $L_1 \ket \psi \perp \ket \psi$,
the inital value of $\dot \F$ saturates
Eq.~\eqref{eq_dflj}. As a specific case,
we can consider a damped harmonic oscillator,
starting in a Fock state $\ket n$,
with $L_1 = g \sqrt{\gamma} a$, so $L_1 \ket n
\propto \ket{n-1}$. Since the rate of damping-induced transitions
is immediately non-zero (and only decreases with time),
to determine the strength of the damping,
we cannot do better than checking for decays
over many short time periods (with intermediate resets).
Since the environment is taken to be Markovian, with vanishing
coherence time, there is no quantum Zeno effect, and
we cannot enhance the damping rate by building
up correlations with the detector system.

Therefore, in circumstances like
these, the parameter estimation story is 
considerably less complicated than for Hamiltonian
parameter estimation. Even though it is possible in principle for $g$-dependent Lindblad operators to yield more complex behaviour---for instance, as in the $g$-dependent `dephasing direction' example (Eq.~\eqref{eq_gdirection}) discussed above---this does not seem
to commonly arise in cases of physical interest.

\bibliography{metrology}

\end{document}